\documentclass{article}

\usepackage[affil-it]{authblk}
\usepackage{graphicx,bm}
\usepackage{color,bbm,subfig,caption,bm,fontenc,currvita}

\usepackage{amsmath,amsthm}
\usepackage{amsfonts}
\usepackage{amssymb,color,makeidx,enumerate,bbm}

\makeatletter
\renewcommand{\pod}[1]{\allowbreak\mathchoice
  {\if@display \mkern 18mu\else \mkern 8mu\fi (#1)}
  {\if@display \mkern 18mu\else \mkern 8mu\fi (#1)}
  {\mkern4mu(#1)}
  {\mkern4mu(#1)}
}
\newtheorem{theorem}{Theorem}[section]
\newtheorem{ther}{Theorem}[section]

\newtheorem{example}[theorem]{Example}

\newtheorem{proposition}[theorem]{Proposition}
\newtheorem{prop}{Proposition}[section]
\newtheorem{definition}[theorem]{Definition}
\newtheorem{defi}{Definition}[section]
\newtheorem{lemma}[theorem]{Lemma}
\newtheorem{remark}[theorem]{Remark}
\newtheorem{rem}{Remark}[section]
\newtheorem{corollary}[theorem]{Corollary}
 
\newtheorem{lem}{Lemma}[section]
\newtheorem{cor}{Corollary}[section]

\newcommand{\koop}{\mathcal{K}}

\def\x{{\bf x}}
\def\n{{\bf n}}
\def\y{{\bf y}}
\def\z{{\bf z}}
\def\s{{\bf s}}

\def\w{{\bf w}}
\def\a{{\bf a}}

\def\k{{\bf k}}
\def\m{{\bf m}}
\def\bp{{\bf p}}
\def\cB{{\cal{B}}}
\def\cD{{\cal{D}}}
\def\cL{{\cal{L}}}

\def\ba{\backslash}

\newcommand{\cH}{\mathcal{H}}
\newcommand{\caA}{\mathcal{A}}

\newcommand{\cl}{\mathrm{cl}}

\newcommand{\ve}[1]{\mathbf{#1}}

\newcommand{\mbf}{\boldsymbol} 

\def\A{{\bf A}}
\def\F{{\bf F}}
\def\G{{\bf G}}

\def\S{{\bf S}}

\def\v{{\bf v}}

\def\u{{\bf u}}
\def\f{{\bf f}}
\def\g{{\bf g}}
\def\h{{\bf h}}
\def\c{{\bf c}}
\def\p{\partial}
\def\bZ{\mathbb{Z}}

\def\bN{\mathbb{N}}
\def\bC{\mathbb{C}}
\def\C{\mathbb{C}}
\def\1{\mathbbm{1}}

\def\bR{\mathbb{R}}
\def\R{\mathbb{R}}
\def\bea{\begin{eqnarray}}
\def\eea{\end{eqnarray}}
\def\be{\begin{equation}}
\def\ee{\end{equation}}
\def\eps{\epsilon}

\def\bomega{\bm{\omega}}
\def\btheta{\bm{\theta}}
\def\bphi{\bm{\phi}}
\def\bOmega{\mbf{\Omega}}
\def\rar{\rightarrow}
\usepackage{layout}

\begin{document}

\title{\color{black} Spectrum of the Koopman Operator, Spectral Expansions in Functional Spaces, and State Space Geometry }

\date{\today}
\author{Igor Mezi\'c}
\affil{Mechanical Engineering and Mathematics, University of California Santa Barbara,
Santa Barbara, CA 93106}
\maketitle
\tableofcontents
\newpage
\begin{abstract}
We examine spectral operator-theoretic properties of linear and nonlinear dynamical systems with globally stable attractors. Using the Kato Decomposition we develop a spectral expansion for general linear autonomous dynamical  systems with analytic observables, and define the notion of generalized eigenfunctions of the associated Koopman operator. We  interpret stable, unstable and center subspaces in terms of zero level sets of generalized eigenfunctions. We then utilize conjugacy properties of Koopman eigenfunctions and the new notion of open eigenfunctions - defined on subsets of state space - to extend these results to nonlinear dynamical systems with an equilibrium. We provide a characterization of (global) center manifolds, center-stable and center-unstable manifolds in terms of joint zero level sets of  families of Koopman operator eigenfunctions associated with the nonlinear system. {\color{black} After defining a new class of Hilbert spaces, that capture the on and off attractor properties of dissipative dynamics, and introduce the concept of Modulated Fock Spaces}, we develop spectral expansions for a class of dynamical systems possessing globally stable  limit cycles and 
limit tori, with observables that are square-integrable in on-attractor variables and analytic in off-attractor variables.  We  discuss definitions of stable, unstable and global center manifolds in such nonlinear systems with (quasi)-periodic attractors in terms of zero level sets of Koopman operator eigenfunctions. We define the notion of isostables for a general class of nonlinear systems. In contrast with the systems that have discrete Koopman operator spectrum, we provide a simple example of a measure-preserving system that is not chaotic but has continuous spectrum, and discuss experimental observations of spectrum on such systems. {\color{black} We also provide a brief characterization of the data types corresponding to the obtained theoretical results and define the coherent principal dimension for a class of datasets based on the lattice-type principal spectrum of the associated Koopman operator.}
\end{abstract}
\section{Introduction}
Spectral theory of dynamical systems shifts the focus of investigation of dynamical systems behavior away from trajectories in the state space 
and towards spectral objects - eigenvalues, eigenfunctions and eigenmodes - of an associated linear operator. Specific examples are the Perron-Frobenius operator \cite{LasotaandMackey:1994} and the composition operator  - in measure-preserving setting called the Koopman operator \cite{Koopman:1931,SinghandManhas:1993}. In this paper we study spectral properties of the composition operator for a class of dynamical systems and relate them to state space and data analyses. 

In classical dynamical systems theory, the notion of conjugacy is an important one. For example, conjugacy is the setting in which linearization theorems, such as the Hartman-Grobman theorem, are proved. In the original investigations using the operator-theoretic approach to measure-preserving dynamical systems, the notion of conjugacy also played an important role \cite{vonNeumann:1932}. One of the most important questions in that era was whether spectral equivalence  of the Koopman operator spectra implies conjugacy of the associated dynamical systems. It was settled in the negative by von Neumann and Kolmogorov \cite{RedeiandWerndl:2012}, where the examples given had complex - mixed or continuous - spectra. The transformation of spectral properties under conjugacy, pointed out in \cite{Budisicetal:2012}, was already used in the data-analysis context in \cite{Williamsetal:2015}. Here we explore the relationship between the spectrum of the composition operator and conjugacy, for dissipative systems, and discuss the type of spectrum they exhibit for asymptotic behavior ranging from equilibria to quasi-periodicity. The approach, inspired by ideas in \cite{LanandMezic:2013} extends the analysis in that paper to provide spectral expansions and treat the case of saddle point equilibria using the newly defined concept of open eigenfunctions of the Koopman operator on subsets of state space. While these systems have discrete spectrum, we also present a simple example of a measure-preserving (non-dissipative) system with non-chaotic dynamics with continuous spectrum.

In dissipative systems, the composition operator is typically {\it non-normal}, and can have {\it generalized eigenfunctions}. 
Gaspard and collaborators studied spectral expansions for dynamical systems containing equilibria and exhibiting pitchfork and Hopf bifurcations \cite{Gaspardetal:1995,GaspardandTasaki:2001}.  
 The author presented the general on-attractor version of the expansion for evolution equations (possibly infinite-dimensional) possessing a finite-dimensional attractor in \cite{Mezic:2005}. 
It is important to note that spectra of dynamical systems can have support on non-discrete sets in the complex plane, provided the 
space of observables is large enough, or the dynamics is complex enough \cite{Gaspardetal:1995,Mezic:2005}. Here, we restrict our attention largely to observables that are $L^2$ in on-attractor variables and analytic in off-attractor variables, and find that the resulting spectra are - for quasi-periodic systems - supported on discrete sets in the complex plane. This observation by the author lead to development of the analytic framework for dissipative dynamical systems using Hardy-type spaces for dynamical systems, in which the composition operator is always spectral \cite{MohrandMezic:2014}. {\color{black}  Here we extend that analysis to provide a new, Hilbert space setting for spectral analysis of dissipative dynamical systems. The resulting spaces are tensor products of spaces suitable for on-attractor dynamics and off attractor dynamics. We prove that the spectrum of the Koopman operator on these spaces is the closure of the product of the ``on-attractor" and ``off-attractor" spectra, and apply these ideas to study systems with limit cycle and quasi-periodic attractors. We introduce two new types of spaces: the modulated Fock space that has properties of the  Fock (or Fock-Bargmann, or Segal-Bargmann) space, but is defined with respect to principal eigenfunctions of the Koopman operator, and the Averaging Kernel Hilbert Space (AKHS) which is a modification of the concept of the Reproducing Kernel Hibert Space (RKHS), the class that modulated and regular Fock spaces belongs to. There are a number of publications on spectrum of composition operators for dissipative dynamical systems that pursue spectral analysis in Hilbert space setting. Fock space has been used specifically in \cite{Bargmann:1962,NewmanandShapiro:1966,carswelletal:2003}.  But these and other works, in different analytic function spaces, are all restricted to the case when the dynamical system has an attracting fixed point, and there is no need for the tensor product construction (see e.g. \cite{bandtlowetal:2017,cowenandmaccluer:1995}).}

Eigenfunctions of the composition operator contain information about geometry of the state space. For example, invariant sets \cite{Mezic:1994},
isochrons \cite{MezicandBanaszuk:2004,MauroyandMezic:2012} and isostables \cite{Mauroyetal:2013}, can all be defined as level sets of eigenfunctions of the operator. Here we extend this set of relations by showing that center-stable, center and center-unstable manifolds of an
attractor can be defined as joint $0$-level sets of a set of eigenfunctions. This can be viewed as shifting the point of view on such invariant manifolds from local - where the essential ingredient of their definition  is tangency to a linear subspace \cite{GuckenheimerandHolmes:2002} - to a global, level-set based definition. The connections between geometric theory and operator theory are explored further here: Floquet analysis in the case of a limit cycle, and generalized Floquet analysis \cite{Sell:1981} in the case of limit tori are used to obtain {\it global} (as opposed to local, as in geometric theory) results on spectral expansions. The usefulness of spectral expansions stems from the fact that most contributions to dynamics of a typical autonomous dissipative systems are exponentially fast, and the dynamics is taken over by the slowest decaying modes and zero-real part eigenvalue modes. This has relationship to the theory of inertial manifolds.

On the data analysis side, the operator-theoretic analysis has recently gained popularity in conjunction with numerical methods such as variants of the Dynamic Mode Decomposition (DMD), Generalized Laplace Analysis,  Prony and Hankel-DMD analysis, as well as compactification methods \cite{MezicandBanaszuk:2004,Mezic:2005,Rowleyetal:2009,Budisicetal:2012,Williamsetal:2015,Bruntonetal:2016,Giannakisetal:2015,SusukiandMezic:2015,ArbabiandMezic:2016}, that can approximate part of the spectrum of an underlying linear operator under certain conditions on the data structure \cite{ArbabiandMezic:2016}. Since these methods operate directly on data (observables), they have been used to analyze a large variety of dynamical processes in many applications. We classify here the types of spectra  associated with dynamical systems of different transient and asymptotic behavior, including systems with fixed point, limit cycle and quasiperiodic attractors. The spectrum is always found out to be of what we call the {\it lattice type}, and is defined as a linear combination over integers of $n$ {\it principal eigenvalues}, where $n$ is the dimension of the state space. This can help with understanding the dynamics underlying the  spectra obtained from data. Namely, the {\it principal coherent dimension} of the data can be determined by examining the lattice and finding the number of {\it principal eigenvalues}.

The paper is organized as follows: in section \ref{sect:lin} we consider the case of linear systems, including those for which geometric and algebraic multiplicity is not equal. We obtain the spectral expansion using the Kato Decomposition. We also obtain  explicit generalized eigenfunctions of the associated composition operator. Using the spectral expansion, the stable, unstable and center subspaces are defined as joint zero level sets of collections of eigenfunctions. An extension of these ideas to nonlinear systems with  equilibria is given in section \ref{ne}, utilizing developments on conjugacy and spectrum in section \ref{sect:conj}, and the new concept of open eigenfunctions of the Koopman operator. The linearization theorem of Palmer is used to provide global definitions of center, center-stable and center-unstable manifolds using zero level sets of collections of composition operator eigenfunctions. {\color{black} The discussion of Hilbert spaces of interest in spectral Koopman operator framework is provided in section \ref{sect:hilb}, where Modified Fock Spaces are defined}.
Spectral expansion theorems for asymptotically limit cycling systems are given in section \ref{sect:spectlc2D} for 2D systems and in section \ref{sect:lc} for n-dimensional systems with a limit cycle. The reason for distinguishing between these two cases is that in the 2-dimensional case the eigenfunctions and eigenvalues can be derived explicitly in terms of averages over the limit cycle, while in the general case we use Floquet theory, due to the non-commutativity of linearization matrices along the limit cycle. For both of these cases we define the concept of Averaging Kernel Hibert Space, in which the Koopman operator is spectral. In section \ref{sect:lt} we derive the spectral expansion for systems globally stable to a limit torus, where attention has to be paid to the exact nature of the dynamics on the torus. Namely, Kolmogorov-Arnold-Moser type Diophantine conditions are needed for the asymptotic dynamics in order to derive the spectral expansion, providing another nice connection between the geometric theory and the operator theoretic approach to dynamical systems. We discuss the possibility of determining the {\it principal coherent dimension} of the data using spectral expansion results in section \ref{sec:data}. In section \ref{cont}  we present a measure-preserving system that has a continuous Koopman operator spectrum, but integrable dynamics, and discuss the consequence for data analysis in such systems. We conclude in section \ref{sect:conc}.
{\color{black} \section{Preliminaries}
For a dynamical system 
\be 
\dot \x=\F(\x),
\label{DSGen}
\ee
 defined on a state-space $M$ (i.e. $\x\in M$ - where we by slight abuse of notation identify a point in a manifold $M$ with its vector representation $\x$ in $\R^m$, $m$ being the dimension of the manifold), where $\x$ is a vector and $\F$ is a possibly nonlinear vector-valued smooth function, of the same dimension as its argument $\x$, denote by $\S^t(\x_0)$ the position at time $t$ of trajectory of   (\ref{DSGen}) 
 that starts at time $0$ at point $\x_0$ (see Figure \ref{traj}). We  call $\S^t(\x_0)$ the flow.
\begin{figure}[ht]
\centering
\fbox{\includegraphics[
natheight=3.551800in, natwidth=6in, height=2.348in, width=3in
]{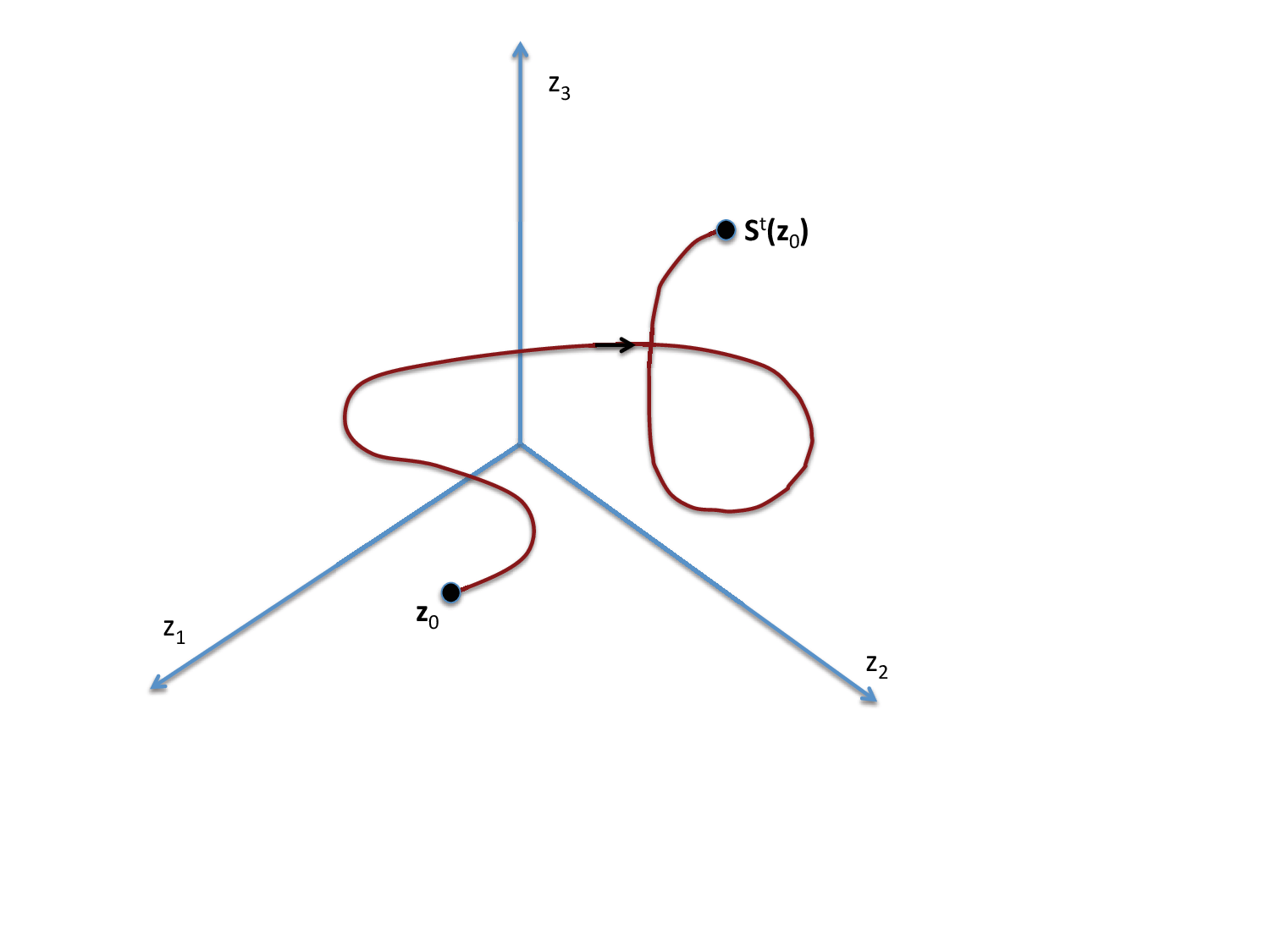}} \caption{Trajectory of a dynamical system in $\mathbb{R}^3$.}
\label{traj}
\end{figure}
  Denote by $\g$ an arbitrary, vector-valued observable from $M$ to $\mathbb{R}^k$.  The value of this observable $\g$ that the system trajectory starting from $\x_0$ at time $0$ sees at time $t$ is 
 \be
 \g(t,\x_0)=\g(\S^t(\x_0)).
 \ee
Note that the space of all observables $\g$ is a linear vector space.  The family of operators $U^t,$ acting on the space of observables parametrized by time $t$ is defined by 
 \be
 U^t\g(\x_0)=\g(\S^t(\x_0)).
 \label{Koopdef}
 \ee
Thus, for a fixed time $\tau$, $U^\tau$ maps the vector-valued observable $\g(\x_0)$ to $\g(\tau,\x_0)$. We will call the family of operators $U^t$ indexed by time $t$ the Koopman operator of the continuous-time system (\ref{DSGen}).
This family was defined for the first time in \cite{Koopman:1931}, for Hamiltonian systems. In operator theory, such operators, when defined for general dynamical systems,  are often called composition operators\index{Composition operators}, since $U^t$ acts on observables by composing them with the mapping $\S^t$ \cite{sm:1993}. 
}
\section{Linear systems}
{\color{black} In this section we study linear dynamical systems from the perspective of spectral theory of Koopman operator family. We include the results previously announced in \cite{Mezic:2015}, but expand on the details of the proofs. Section \ref{sec:can} is new and utilizes Koopman eigenfunctions as new coordinates that transform the linear dynamical system into its canonical (Jordan) form. This is used later in the paper where conjugacy theorems are used for identification of Koopman eigenfunctions.}
\label{sect:lin}
\subsection{Continuous-time Linear Systems with Simple Spectrum}
In the case when the dynamical system is linear, and given by $\dot \x=A\x,\ \ \x\in \R^n$ its matrix eigenvalues are eigenvalues of the associated Koopman operator.  The associated Koopman eigenfunctions are given by \cite{Rowleyetal:2009}:
\begin{equation}
	\phi_j(\x) = \left< \x, \w_j\right>,\qquad j=1,\ldots,n
\end{equation}
where  $\w_j$ are  eigenvectors of the adjoint $A^*$ (that is,
	$A^* \w_j = \lambda^c_j \w_j$),
normalized so that $\left< \v_j,\w_k \right> = \delta_{jk}$, where $\v_j$ is an eigenvector of $A$, and $\left<\cdot,\cdot\right>$ denotes an inner product on the linear space $M$
in which the evolution is taking place. This is easily seen by observing
\be
\dot \phi_j= \left< \dot\x, \w_j\right>= \left< A\x, \w_j\right>= \left< \x, A^*\w_j\right>=\lambda_j \left< \x, \w_j\right>= \lambda_j \phi_j,
\ee
and thus
$
\phi_j(t,\x_0)=U^t\phi_j(\x_0)=\exp(\lambda_j t)\phi_j(\x_0).
$
Now, for any $\x\in M$, as long as~$\A$ has a full set of eigenvectors at distinct eigenvalues $\lambda_j$, we may write
$$
\x = \sum_{j=1}^n \left<\x,\w_j\right>\v_j = \sum_{j=1}^n \phi_j(\x) \v_j.
$$
Thus,
\begin{eqnarray}
	   U^t\x(\x_0)&=& \x(t,\x_0)= \exp(At)\x_0=\sum_{j=1}^n   \left<\exp(A t)\x_0,\w_j\right> \v_j\nonumber \\
	     &=&\sum_{j=1}^n   \left<\x_0,\exp(A^* t)\w_j\right> \v_j=\sum_{j=1}^n \exp(\lambda_j t) \left<\x_0,\w_j\right>\v_j,\nonumber \\
	     &=&\sum_{j=1}^n \exp(\lambda_j t)\phi_j(\x_0) \v_j,
\label{eq:x_expansion}
\end{eqnarray}
where  $\x(\x_0)$ is the vector function that associates Cartesian coordinates with the point $\x_0$ (the initial condition) in state space. 
This is an expansion of the dynamics of observables - in this case the coordinate functions $\x(\x_0)$ in terms of spectral quantities (eigenvalues, eigenfunctions and {\it Koopman modes} $\v_j$) of the Koopman family $U^t$. Considering (\ref{eq:x_expansion}), we note that the quantity we know as the eigenvector $\v_j$ is {\it not} associated with the Koopman operator, but rather with the observable - if we changed the observable to, for example $\y=C\x$, $C$ being an $m\times n$ matrix, then the expansion would read 
\be
	   U^t\y(\x_0)=\sum_{j=1}^n \exp(\lambda_j t)\phi_j(\x_0) C\v_j,
\label{eq:x_expansiony}
\ee
and we would call $C\v_j$ the $j-th$ Koopman mode\footnote{Koopman modes are defined up to a constant, the same as eigenvectors. However, here we have defined projection with respect to a specific basis with an orthonormal dual.} of observable $\y$. Assume now that the space of observables on $\R^n$ we are considering is the space of complex linear combinations of $\x(\x_0)$. Then, $\phi_j(\x_0)C\v_j$  is the projection of the observable $C\x$ onto the eigenspace of the Koopman family spanned by the eigenfunction $\phi_j(\x_0)=\langle \x,\w_j\rangle$. 

Note that what changed between expansions (\ref{eq:x_expansion}) and (\ref{eq:x_expansiony}) is the Koopman modes. On the other hand, the eigenvalues and eigenfunctions used in the expansion do not change. 
Thus, what changes with change in observables  is their contribution to the overall evolution in the observable, encoded in $C\v_j$.
These properties persist in the fully nonlinear case, with the modification that the spectral expansion is typically infinite and can have a continuous spectrum part.

Note also that the evolution of coordinate functions can be written in terms of the evolution of Koopman eigenfunctions, by
\be
 U^t\x(\x_0)=\sum_{j=1}^n \phi_j(t,\x_0) \v_j.
\label{evoK}
\ee

\subsection{Continuous-time Linear Systems: the General Case}
\label{sec:lingen}
In general, the matrix $A$ can have repeated eigenvalues and this can lead to a lack of eigenvectors.
Recall that the algebraic multiplicity of an eigenvalue $\lambda_j$ of $A$ is the exponent ($m_j$) of the polynomial factor  $(\lambda-\lambda_j)^{m_j}$
of the characteristic polynomial $\det(A-\lambda I)$. In other words, it is the number of repeat appearences of $\lambda_j$ as a zero of the characteristic polynomial. An eigenvalue that repeats $m_j$ times does not necessarily have $m_j$ eigenvectors associated with it. Indeed - the algebraic multiplicity $m_j$  of $\lambda_j$ is bigger than or equal to geometric multiplicity, which is the number of eigenvectors associated with $\lambda_j$.
Such sonsiderations lead to the so-called Kato Decomposition. Kato Decomposition is an example of a spectral decomposition, where a linear operator is decomposed into a sum of terms consisting of scalar multiples of projection and nilpotent operators. For a finite-dimensional linear operator $A$ it reads \cite{Kato:2013}:
\be
U=\sum_{h=1}^s \lambda_hP_h+D_h,
\label{rep}
\ee
where $s$ is the dimension. Each $P_h$ is a projection operator on the {\it algebraic eigenspace} $M_h$ that can be defined as the null space of $(U-\lambda_hI)^{m_h}$, and $D_h$ is a nilpotent operator.
We now use this spectral decomposition theorem for finite-dimensional 
linear operators to provide an easy, elegant proof of Hirsch-Smale theorem \cite{HirschandSmale:1974} on solutions of 
ordinary differential equations. 
Consider a linear ordinary differential equation on $\bR^m,$
$
\dot \x=A\x
$
where $A$ is an $n\times n$ matrix. It is well-known that the solution of this equation
reads 
$
\x(t)=\exp(At)\x_0,
$
where $\x_0$ is the initial condition. The exponentiation of the matrix $A$ reads
\be
\exp(At)=\sum_{k=0}^{\infty}\frac{A^kt^k}{k!}.
\label{exp}
\ee
Now, from the  Kato decomposition, and using the fact that
\be
D_hD_k=\delta_{hk}D_h,P_hD_k=D_hP_k=0,
\ee
 we obtain
\be
A^k=\sum_{h=1}^s \lambda_h^kP_h+\sum_{j=1}^k\binom{k}{j}\sum_{h=1}^s \lambda_h^{k-j}D_h^j,
\ee
where $\lambda_h,h=1,...,s$ are eigenvalues of $A$.
We rewrite $\exp(At)$ as
\bea
&&I+\sum_{k=1}^{\infty}\frac{\sum_{h=1}^s \lambda_h^kP_h+\sum_{j=1}^k\binom{k}{j}\sum_{h=1}^s \lambda_h^{k-j}D_h^j}{k!}t^k, \nonumber \\
&=&\sum_{h=1}^s P_h \sum_{k=0}^{\infty}\frac{\lambda_h^kt^k}{k!}+\sum_{k=1}^{\infty}\frac{\sum_{j=1}^k\binom{k}{j}\sum_{h=1}^s \lambda_h^{k-j}D_h^j}{k!}t^k, \nonumber \\
&=&\sum_{h=1}^s e^{\lambda_ht}P_h+\sum_{k=1}^{\infty}\frac{\sum_{j=1}^k\binom{k}{j}\sum_{h=1}^s \lambda_h^{k-j}D_h^j}{k!}t^k, \nonumber \\
\label{exp1}
\eea
Note now that 
\be
t^le^{\lambda_h t}=t^l\sum_{k=0}^{\infty}\frac{\lambda_h^kt^k}{k!}=\sum_{k=0}^{\infty}\frac{\lambda_h^kt^{k+l}}{k!}
=\sum_{m=l}^{\infty}\frac{\lambda_h^{m-l}t^m}{(m-l)!}.
\ee
We can rewrite the second sum in the last line of (\ref{exp1})
as
\be
\sum_{h=1}^s\sum_{j<m_h}\sum_{k=j}^\infty\binom{k}{j} \lambda_h^{k-j}D_h^j\frac{t^k}{k!}, \\
\ee
leading further to
\bea
&=&\sum_{h=1}^s\sum_{j<m_h}\sum_{k=j}^\infty\frac{k\cdot(k-1)\cdot...\cdot(k-j+1)}{j!} \lambda_h^{k-j}D_h^j\frac{t^k}{k!},\nonumber \\
&=&\sum_{h=1}^s\sum_{j<m_h}\frac{D_h^j}{j!}\sum_{k=j}^\infty \lambda_h^{k-j}\frac{t^k}{(k-j)!},\nonumber \\
&=&\sum_{h=1}^s\sum_{j<m_h}\frac{D_h^j}{j!}t^je^{\lambda_ht}.\\
\eea
Thus we get
\be
\exp(At)=\sum_{h=1}^s ( e^{\lambda_ht}P_h+\sum_{j<m_h}\frac{t^je^{\lambda_ht}}{j!}D_h^j), 
\label{exp2}
\ee
Let us now connect this expansion to the formula we obtained previously, given by (\ref{eq:x_expansion}). In that case, we assumed that algebraic multiplicities of all eigenvalues are $1$, and there is a full set of associated eigenvectors $\v_h$. Thus, the nilpotent part $D_h=0$, and the projection 
of a vector $\x_0$ on the $h-th$ eigenspace is 
\be
P_h {\x_0}=\left<\x_0,\w_h\right>\v_h=\phi_h(\x_0)\v_h.
\ee
Using this with (\ref{exp2}), we obtain (\ref{eq:x_expansion}). 

More generally, let the dimension of each geometric eigenspace be equal to $1$, let $j=1,...,s$ be the counter of distinct eigenvalues of $A$ and $m_1,...,m_s$ their multiplicities (or equivalently dimensions of algebraic eigenspaces corresponding to eigenvalues). Label the basis of the generalized eigenspace $E_h$ by $\v_h^1,...,\v_h^{m_h}$, where 
 $\v_h^i$ are chosen so that $(A-\lambda_hI)^i\v_h^i=0$. In other words, $\v_h^1$ is a standard eigenvector of $A$ at $\lambda_h$ and the generalized eigenvectors $\v_h^i,i=2,...,m_h$ satisfy
 $
A\v_h^i=\lambda_h\v_h^i+\v_h^{i-1}.
$
 Now let 
$
\phi_h^i(\x)=\left<\x,\w_h^i\right>
$
where $\w_h^i$ is the dual basis vector to $\v_h^i$ and satisfies 
\bea
A^*\w_h^i&=&\lambda_h^c\w_h^i+\w_h^{i+1}, \ i<m_h. \nonumber \\
A^*\w_h^i&=&\lambda_h^c\w_h^i, \ \ \ \ \ \ \ \ \ \ \      i=m_h.
\eea 
Note that for $i>1$.
\bea
\dot\phi_h^i(\x)&=&\left<\dot \x,\w_h^i\right>=\left<A \x,\w_h^i\right> \nonumber \\
&=&\left< \x,A^*\w_h^i\right>=\left< \x,\lambda_h^c\w_h^i+\w_h^{i+1}\right> \nonumber \\
&=&\lambda_h\left< \x,\w_h^i\right>+\left<\x,\w_h^{i+1}\right> \nonumber \\
&=&\lambda_h\phi_h^i(\x)+\phi_h^{i+1}(\x).
\label{geneig}
\eea
We call $\phi_h^i(\x),1\leq i <m_h$   the generalized eigenfunctions of the Koopman operator at eigenvalue $\lambda_h$. 
{\color{black} \begin{rem}
 It is evident from the equation (\ref{geneig}) that products of generalized eigenfunctions are not generalized eigenfunctions, in contrast with the property of ordinary eigenfunctions.
\end{rem}}
\begin{example}
\label{exa:ge}
To justify the name  generalized eigenfunctions, consider the following simple example: let $m_h=2$. Then $\dot \phi_h^1=\lambda_h \phi_h^1+\phi_h^2$, where $\phi_2^1$ is an eigenfunction of $U^t$ at $\lambda_h$ satisfying
$\dot \phi_h^2=\lambda_h \phi_h^2.$
Then 
\be
(d/dt-\lambda_h I)^2 \phi_h^1=0.
\ee
Thus, $\phi_h^1$ is in the nullspace of the differential operator $(\frac{d}{dt}-\lambda_h I)^2$. 
\end{example}
Expanding from Example \ref{exa:ge}, for $m_j$ arbitrary, generalized eigenfunctions $\phi$ satisfy 
$
(d/dt-\lambda_h I)^{m_h} \phi=0.
$
By integrating (\ref{geneig}), the time evolution of the generalized eigenfunctions reads
\be
\phi_h^i(t)=\sum_{n=0}^{m_h-i}\frac{t^{n}}{n!}e^{\lambda_h t}\sum_{l=m_h}^{i+n}\phi_h^l(0).
\label{eigevo}
\ee
(in fact by directly differentiating (\ref{eigevo}), one can easily find out that it satisfies (\ref{geneig})).
Now writing
\be
\x_0=\sum_{h=1}^s\sum_{i=1}^{m_h}\left< \x_0,\w_h^i\right>\v_h^i,
\ee
we get
\begin{eqnarray}
	   U^t\x(\x_0)&=& \x(t)= \exp(At)\x_0\nonumber \\
	   &=&\sum_{h=1}^s\sum_{i=1}^{m_h}   \left<\exp(A t)\x_0,\w_h^i\right> \v_h^i\nonumber \\
\label{eq:x_expansion1}
\end{eqnarray}
leading further to
\bea
&=&\sum_{h=1}^s\sum_{i=1}^{m_h} \left<\x_0,\exp(A^* t)\w_h^i\right> \v_h^i\nonumber \\
	   &=&\sum_{h=1}^s\sum_{i=1}^{m_h}e^{\lambda_h t}(\sum_{k=i}^{m_h} \frac{t^{k-i}}{(k-i)!}\left< \x_0,\w_h^k\right>)\v_h^i,\nonumber \\
&=&\sum_{h=1}^s e^{\lambda_h t} \sum_{i=1}^{m_h}(\sum_{k=i}^{m_h} \frac{t^{k-i}}{(k-i)!}\phi_h^k(\x_0))\v_h^i \nonumber \\
	      &=&\sum_{h=1}^s \left[ e^{\lambda_h t}\left(\sum_{k=1}^{m_h} \phi_h^k(\x_0)\v_h^k \right)\right.\nonumber \\
	      &+&\left.\sum_{i=1}^{m_h-1}\frac{t^{i}}{i!}e^{\lambda_h t}\left(\sum_{k=i+1}^{m_h} \phi_h^k(\x_0)\v_h^{k-i}\right)\right]. 
	      \label{eq:genlinexp}
\eea
We connect the formula we just obtained with the expansion (\ref{exp2}). Comparing the two, it is easy to see that 
\be
P_h\x=\sum_{k=1}^{m_h} \phi_h^k(\x)\v_h^k,
\label{pj}
\ee
and 
\be
D_h^i\x=\sum_{k=i+1}^{m_h} \phi_h^k(\x)\v_h^{k-i}
\label{nilp}
\ee

The above discussion also shows that, as long as we restrict the space of observables on $\bR^m$ to linear ones, $f(\x)=\left<{\bf c},\x\right>$, where $\c$ is a vector in $\bR^m$, then the generalized eigenfunctions and associated eigenvalues of the Koopman operator are  obtainable in a straightforward fashion from the standard linear analysis of $A$ and its transpose.

It is easy to see that the most general case, in which dimension of geometric eigenspaces is not necessarily $1$, is easily treated by considering geometric eigenspace of dimension say $2$ as two geometric eigenspaces of dimension $1$. Keeping in mind that these correspond to - numerically - the same eigenvalue, we can define generalized eigenvectors corresponding to each eigenvector in - now separate - $1$-dimensional geometric eigenspaces.

\subsection{The Canonical Form of Linear Systems}
\label{sec:can}
The (generalized) Koopman eigenfunctions $$\phi_h^i(\x)=\left< \x,\w_h\right>,\quad h=1,...,s\, \quad0\leq i<m_h,$$ can be thought of as ``good" coordinates for linear systems. Let $${\boldsymbol \phi}=(\phi_1^1,...,\phi_1^{m_1},\phi_2^1,...,\phi_2^{m_2},...,\phi_s^1,...,\phi_s^{m_s})^T.$$ From
\be
\dot\phi_h^i(\x)=\lambda_h\phi_h^i(\x)+\phi_h^{i+1}(\x),
\ee
 we obtain 
\be\dot {\mathbf \phi}=J{\mathbf \phi},
\label{eq:jord} 
\ee
where
\be J =
\left(
\begin{array}{*4{c}}
A_{\lambda_1} & & & 0\\
& A_{\lambda_2} & & \\
& & \ddots & \\
0 & & & A_{\lambda_s}
\end{array}
\right)
\ee
and
\be  A_{\lambda_h}=
\left(
\begin{array}{*5{c}}
\lambda_h & 1 & 0 &  \cdots & 0\\
0& \lambda_h & 1 &  \cdots & 0\\
\vdots& & \cdots & & \vdots\\
0 & \ddots & & \lambda_h & 1 \\
0 & \ddots & & 0 & \lambda_h
\end{array}
\right)
\ee
is the Jordan block \index{Jordan block} corresponding to the eigenvalue $\lambda_j$.
 Note that Koopman eigenfunctions can be complex (conjugate) and thus this representation
is in general complex. 

The real form of the Jordan block corresponding to a complex eigenvalue $\lambda_i$ whose geometric multiplicity is
less than algebraic multiplicity is obtained  using the variables $r_j=y_j$ (for $\lambda_j\in\mathbb{R}$) and the polar coordinates 
\be
\left(\begin{array}{c}y_j \\ y_{j+1}\end{array} \right)=\left(\begin{array}{c} r_j \cos(\theta_j)\\ r_j \sin(\theta_j) \end{array} \right)
\label{eq:realcoord}
\ee
 (for $\lambda_j=\lambda_{j+1}^c=r_je^{i\theta_j}\, \notin\mathbb{R}$). Thus,
$$\phi_h=r_h(\x)e^{i\theta_h(\x)}=\phi_{h+1}^c,$$ and $(\phi_h(\x),\phi_{h+1}(\x))$ get transformed into $(y_h(\x),y_{h+1}(\x))$ to yield the $i$-th Jordan block
\be
A_{\lambda_i}= 
\begin{bmatrix}
C_i    & I       & \;     & \;    \\
\;     & C_i     & \ddots & \;    \\     
\;     & \;      & \ddots & I     \\
\;     & \;      & \;     & C_i   \\
\end{bmatrix}.
\ee
where 
\be
C_i = 
\begin{bmatrix}
\sigma_i  & \omega_i \\
-\omega_i & \sigma_i \\ 
\end{bmatrix},
\ee
and $I$ is the $2\times 2$ identity matrix. 

We have the following corollary of the above considerations:
\begin{cor}\label{cor:genef} If a set of functions $\bphi$ satisfy 
\be\dot {\bphi}=J{ \bphi},
\label{eq:jord1} 
\ee
where $J$ is the complex Jordan normal form of a matrix $A$, then $\bphi$ is a set of (generalized) eigenfunctions of 
\be
\dot\x=A\x.
\ee
\end{cor}
\subsection{Stable, Unstable and Center Subspace}
\label{subspaces}
Let us recall the definition of stable, unstable and center subspaces of $\dot \x=A\x, \x\in \R^n$: the { \it stable}\index{Stable subspace} subspace of the fixed point $0$ is the location of all the points in $\R^n$ that go to the fixed point at the origin as $t\rar \infty$. The stable subspace is classically obtained as the span of (generalized) eigenvectors corresponding to eigenvalues of negative real part. In the same way, the{ \it unstable}\index{Untable subspace} subspace of the fixed point $0$ is the location of all the points that go to the fixed point at the origin as $t\rar -\infty$, and is classically obtained as the span of (generalized) eigenvectors corresponding to eigenvalues of positive real part. The { \it center} subspace\index{Center subspace} is usually not defined by its asymptotics (but could be, as we will see that it is the location of all the points in the state space that stay at the finite distance from the origin, or grow slowly (algebraically) as $t \rar \infty$), but rather as the span of (generalized) eigenvectors associated with eigenvalues of zero real part.

Looking at the equation (\ref{eq:x_expansion}), it is interesting to note that one can extract the geometrical location
of stable, unstable and center subspaces from the eigenfunctions of the Koopman operator. We order
eigenvalues $\lambda_j, j=1,..,n$ from the largest  to the smallest, where we do not pay attention to the possible repeat of eigenvalues. Let $s,c,u$ be the number of negative real part eigenvalues, $0$ and positive real part eigenvalues. 
\begin{prop}Let  $\lambda_1,...\lambda_u$ be positive real part eigenvalues, $\lambda_{u+1},...,\lambda_{u+c}$ be $0$ real part eigenvalues, and $\lambda_{u+c+1},...,\lambda_{s}$ be negative real part eigenvalues of a matrix $A$ of an LTI system. Let \be\phi_1,...,\phi_{u+c+s},\ee be the (generalized) eigenfunctions of the associated Koopman operator.
Then the joint level set of (generalized) eigenfunctions
\be L_s=\{\x\in\R^n|\phi_1(\x)=0,...,\phi_{u+c}(\x)=0\},\ee is the stable subspace $E^s$, 
\be L_c=\{\x\in\R^n|\phi_1(\x)=0,...,\phi_{u}(\x)=0,\phi_{u+c+1}(\x)=0...,\phi_{u+c+s}(\x)=0\},\ee is the center subspace $E^c$, and 
\be L_u=\{\x\in\R^n|\phi_{u+c+1}(\x)=0,...,\phi_{u+c+s}(\x)=0\},\ee the unstable subspace $E^u$.
\end{prop}
\begin{proof} Note that setting $\phi_1(\x)=0,...,\phi_{u+c}(\x)=0$ leads to annulation of terms in (\ref{eq:genlinexp}) that are multiplied by $e^{\lambda_jt}$, where $\Re{\lambda_j}\geq0$. Thus, any initial condition $\x$ belonging to $L_s$ has evolution governed by terms that asymptotically converge to $0$ and thus are parts of the stable subspace. Conversely, assume that $\x$ does not belong to $L_s$, but the trajectory starting at it asymptotically converges to $0$. Since $\x$ has non-zero projection $\left<\x,\w_j\right>$ on at least one of the  (generalized) eigenvectors of $A$, that are associated with eigenvalues of non-negative real part, we get a contradiction. The proof for the unstable subspace is analogous. 

Since the center subspace is defined as the span of the (generalized) eigenvectors of $A$ having eigenvalues with zero real part, the initial condition in the center subspace can not have any projection on (generalized) eigenvectors associated with eigenvalues with positive or negative real part, 
and thus $$\phi_1(\x)=0,...,\phi_{u}(\x)=0,\phi_{u+c+1}(\x)=0...,\phi_{u+c+s}(\x)=0.$$ This implies that $\x$ is in $L_c$. Conversely, if $\x\in L_c$ then $\x$ does not have any projection on (generalized) eigenvectors associated with eigenvalues with positive or negative real part, and thus is in $E^c$.
\end{proof} This generalizes nicely to nonlinear systems (see below), in contrast to the fact that the standard definition, where e.g. the unstable space is the span of $\v_1,...,\v_u$ does not. Namely,  even when the system is of the form $$\dot \x=A\x+\eps\f,$$ for $\f$ bounded,  $\f(0)=0$, and $\eps$ small, using the span of eigenvectors we can only show existence of the   unstable, stable and center manifolds that are tangent to the unstable, stable and center subspace $E^u,E^s, E^c$, respectively.

 So, the joint zero level sets of Koopman eigenfunctions define dynamically important geometric objects - invariant subspaces - of linear dynamical systems. This is not an isolated incident.  Rather, in general the {\it level sets of Koopman eigenfunctions}\index{Level sets of Koopman eigenfunctions} reveal important  information about the state space geometry of the underlying dynamical system.
 
\section{Koopman Eigenfunctions Under Conjugacy}
\label{sect:conj}
Spectral properties of the Koopman operator transform nicely under conjugacy, {\color{black} as already shown in \cite{Mezic:2005}}. Here we use the notion of conjugacy defined more generally than in the classical context. In fact, we will  define the notion of {\it factor conjugacy} \index{Factor conjugacy} - coming from the fact that we are combining notions of factors from measure theory \cite{Petersen:1989}, and the topological notion of conjugacy \cite{GuckenheimerandHolmes:2002}. \index{Conjugacy! factor}

 Let $S^t,U_S^t$ be the family of mappings and the Koopman operator associated with $$\dot \x=\F(\x),\ \x\in \R^n,$$ with $m\leq n$ and $T^t,U_T^t$ a family of mappings and the Koopman operator associated with  $$\dot \y=\G(\y),\ \y\in \R^m.$$ Assume that $\phi(\y)$ is an eigenfunction of $U_T^t$ associated with eigenvalue $\lambda$. In addition, let $\h:\R^n\rar \R^m$ be a mapping such that 
 
 \be \h(S^t\x)=(T^t(\h(\x))),\label{eq:conj}\ee 
 i.e. the two dynamical systems are {\it (factor) conjugate}\index{Conjugacy of dynamical systems}.\footnote{This is not the standard notion of conjugacy, since the dimensions of spaces that $\h$ maps between is not necessarily the same, i.e. $m\neq n$ necessarily.} Then we have 
\be
\exp(\lambda t)\phi\circ\h(\x)=\phi(T^t\h(\x))=\phi(\h(S^t(\y)))
=U_S^t(\phi\circ\h(\x)), \label{conj}
\ee
i.e. if $\phi$ is an eigenfunction at $\lambda$ of $U_T^t$, then the composition $\phi\circ\h$ is an eigenfunction of $U_S^t$ at $\lambda$. As a consequence, if we can find a global conjugacy of a nonlinear system  to a linear system, then the spectrum of the Koopman operator can typically be determined from the spectrum of the linearization at the fixed point. We discuss this, and some extensions, in the next section. The classical notion of topological conjugacy
\index{Conjugacy! factor} \index{Topological conjugacy} is obtained when $m=n$ and $\h$ is a homeomorphism (a continuous invertible map whose inverse is also continuous). If $\h$ is a $C^k$ diffeomorphism, then we have a $C^k$ diffeomorphic conjugacy. The notion of factor conjugacy is broader than those classical definitions, and includes the notion of {\it semi-conjugacy}, that is obtained when $\h$ is continuous or smooth, but $m<n$. \footnote{The definition of factor conjugacy can be generalized to include dynamical systems on spaces $M$ and $N$, where $\h:M\rar N$ and $N\neq M$ (see \cite{MezicandBanaszuk:2004}, where such concept was defined for the case $N=S^1,$ indicating conjugacy to a rotation).}

Generalized eigenfuctions are  preserved under conjugation, just like ordinary eigenfunctions: let $m_j$ be the geometric multiplicity of the eigenvalue $\lambda_j$. For $0<i<m_j$ we have (see (\ref{eigevo})):
\be
\sum_{n=0}^{m_j-i}\frac{t^{n}}{n!}e^{\lambda_j t}\phi_j^{n+i} \circ \h(\x)\nonumber \\
=U_T^t\phi_j^i(\h(\x))
=\phi_j^i(T^t(\h(\x)))=\phi_j^i(\h(S^t\x))
=U_S^t\left(\phi_j^i \circ\h(\x)\right)
\label{e}
\ee
thus indicating that  $\phi_j^i\circ \h$ is a function that evolves in time according to  the evolution equation (\ref{eigevo}) and thus is  a generalized eigenfunction.
Together with the fact that we already proved this for ordinary eigenfunctions in (\ref{conj}), we get 
\begin{prop} [\color{black} \cite{Mezic:2015}]
Let $S^t,U_S^t$ be the family of mappings and the Koopman operator associated with $\dot \x=\F(\x),\x\in \R^n$ and $T^t,U_T^t$ the family of mappings and the Koopman operator associated with  $\dot \y=\G(\y),\y\in \R^n$. In addition, let $\h: \R^n\rar  \R^n$ be a $C^2$ diffeomorphism  such that $ \h(S^t\x)=T^t(\h(\x))$, i.e. the two dynamical systems are $C^2$ diffeomorphically conjugate.  If $\phi$ is a (generalized) $C^2$ eigenfunction at $\lambda$ of $U_T^t$, then the composition $\phi\circ\h$ is a (generalized) $C^2$ eigenfunction of $U_S^t$ at $\lambda$. 
\label{prop:conj}
\end{prop}

\section{Nonlinear Systems with Globally Stable Equilibria }
\label{ne}
 Non-degenerate linear systems (i.e. those with $\det A\ne 0$)
have a single equilibrium at the origin as the distinguished solution. As the natural first extension to the nonlinear realm, it is interesting to consider a class of nonlinear systems that (at least locally) have an equilibrium as the only special solution, and consider what spectral theory of the Koopman operator for such systems can say. {\color{black} In all the work that follows, we assume global existence and uniqueness of solutions of the underlying ordinary differential equations.}

For systems that are stable to an equilibrium from an open attracting set, we develop in this section a theory that strongly resembles that of linear systems - as could be expected once it is understood how Koopman eigenfunctions change under conjugacy.  Geometric notions that were discussed in the previous LTI context, such as stable, unstable and center manifolds  are developed in this section for nonlinear systems with globally stable equilibria. Since we use local conjugacy theorems, such as the Hartman-Grobman theorem, we start with the results that enables extension of an eigenfunction of the Koopman operator from an open set to a larger domain in state space. We will assume $S^t$ is a continuously differentiable flow of $\dot \x=\f(\x)$ defined on a state space $M$.

\subsection{Eigenfunctions of the Koopman Operator Defined on Subsets of the State-Space}
The classical linearization theorems that we will utilize in our study are commonly defined on a neighborhood of 
a set of special dynamical significance, such as an equilibrium point,  an invariant torus, or a strange attractor. The idea we pursue here is that extensions of such ``local" eigenfunctions can be done using the flow, as long as the resulting set in state space does not begin to intersect itself. We first define the notion of  an open eigenfunction and subdomain eigenfunction.
\begin{defi} Let  $\phi:A\rar \C$, where $A\subset M$ is not an invariant set. Let $\x\in A$, and $\tau\in(\tau^+(\x),\tau^{-}(\x))=I_\x,$ a connected open interval such that  $S^\tau(\x)\in A, \ \forall \tau\in I_\x$. If
\be
U^\tau\phi(\x)=\phi(S^\tau(\x))=e^{\lambda \tau} \phi(\x),\  \forall \tau\in I_\x,
\label{eq:effor}
\ee
Then $\phi$ is called an {\bf open eigenfunction} of the Koopman operator family $U^t,t\in \R$, associated with an eigenvalue $\lambda.$

If $A$ is a proper invariant subset of $M$ (in which case $I_{\x}=\R$, for every $\x\in A$), we call $\phi$ a {\bf subdomain eigenfunction}.
\end{defi}
Clearly, if $A=M$ and $I_\x=\R$, for every $\x\in M,$ then $\phi$ is an ordinary eigenfunction.

We now define a function that will help us extend the definition of $\phi$ from $A$ to a larger reachable set when $A$ is open:
\begin{defi} Let $A\subset M$ be an open, connected set in the state space. For $\z\notin A$, let $t(\z)$ to be the time such that  $S^{t(\z)}\z=\x\in \cl(A),$ defined by
\[
t(\z)= \left\{\begin{array}{c} 0,\ \forall \z\in A \\
                                           t\in \R^+,\  \forall \z\in S^{-t}(\cl(A))\ba \cup_{0<\tau<t}S^{-\tau}(A)\\
                                            t\in \R^-,\  \forall \z\in S^{-t}(\cl(A))\ba \left(\cup_{t<\tau<0}S^{-\tau}(A) \cup A\cup_{\gamma\in\R^+}S^{-\gamma(A)}\right).
                                            \end{array} 
                                            \right.
\] 
Let  $$ B=\cup_{t\in \R^+}S^{-t}(\cl(A))\ba \cup_{0<\tau<t}S^{-\tau}(A). $$
Also, let $$ F=\cup_{t\in \R^-}S^{-t}(\cl(A))\ba \left(\cup_{t<\tau<0}S^\tau(A) \cup A\cup_{\gamma\in\R^+}S^{-\gamma(A)}\right). $$ Let $ P\subset M$ be the set of points for which $t(\z)$ is defined, i.e. $P=B\cup A\cup F$. Let $\tau^+(\z)\in\R^+$ and $\tau^-(\z)\in \R\ba \R^+$  be the times such that 
\bea 
\tau^-(\z)<t<\tau^+(\z) &\implies& S^{-t}(\z) \in P \\
t>\tau^+(\z) &\implies& S^{-t}(\z) \notin B,
\eea 
and 
\bea 
t<\tau^-(\z) &\implies& S^{-t}(\z) \notin F
\eea 
Then  $I_\z=(\tau^-(\z),\tau^{+}(\z))$ is a  connected open interval.

\end{defi}
\begin{rem}
One can think of the function $t(\z)$ as the time to enter the closure of $A$, either by a forward flow, or by backward flow from $\z$. $\tau^+$ and $\tau^-$ are the times within which the backward image of the set $A$ and the forward image do not intersect. The set $B$ is the backwards (in time) image of $A$ under the flow, and the set $F$ is the forward image of $A$ under the flow.
\end{rem}
 The following lemma enables 
an extension of eigenfunctions of the composition operator to a larger set:
\begin{lemma} 
\label{lem:ext} 
Let $A\subset M$ be an open, connected set in the state space such that $\phi:A\rar\C$ is a continuous function  that satisfies
\be
\dot\phi(\x)=\lambda\phi(\x),\ \x\in A
\ee
for some $\lambda\in\C$.   For $\z\in P$ define
\be \boxed{
\phi(\z)=e^{-\lambda t(\z)}\phi(S^{t(\z)}\z).
}
\label{eq:defef}
\ee
Assume there is a $\z\in P$ such that $\tau^+(\z)<\infty$, or $\tau^{-}(\z)>-\infty.$ Then,
$\phi$
is a continuous, open eigenfunction of $U^t$ on  $P$ associated with the eigenvalue $\lambda$. 
\end{lemma}
\begin{proof} 

All $I_\z$ contain $0$ and are open and connected.
 Pick a $\z$ in $ P$. For any $\tau\in \R$ such that $\tau \in I_\z$  we have $$t(S^\tau\z)=t(\z)-\tau.$$
We obtain 
\bea
\phi(S^\tau\z)&=&e^{-\lambda t(S^\tau\z)}\phi(S^{t(S^\tau\z)}S^\tau\z) \nonumber \\
&=& e^{-\lambda (t(\z)-\tau)}\phi(S^{t(\z)-\tau}S^\tau\z)
 \nonumber \\
&=& e^{\lambda\tau}e^{-\lambda t(\z)}\phi(S^{t(\z)}\z)
 \nonumber \\
&=& e^{\lambda\tau}\phi(\z),
\eea
and by assumption we know that $P$ is not invariant, so the function defined in (\ref{eq:defef}) satisfies the requirement (\ref{eq:effor}) for being  an open eigenfunction of $U^t$.

 {\color{black} The continuity of $\phi$  is proved as follows: for any given $\eps$ and $\x \in P$ we need to find a $\delta$ such that $|\phi(\x)-\phi(\y)|<\eps$ if $|\x-\y|<\delta$. Start with $\y\in P$ such that for some $t$ with $S^{-t}(\x)\in A$, it holds that  $|S^{-t}(\x)-S^{-t}(\y)|<\delta_1$ for some sufficiently small $\delta_1$, ensuring $|\phi(S^{-t})(\x)-\phi(S^{-t}(\y))|<\eps_1$, and $S^{-t}(\y)\in E\subset A$, $E$ open. Such a set of $\y$ exists since $A$ is open. Now observe that 
 \be
 |\phi(\x)-\phi(\y)|=|e^{\lambda t}(\phi(S^{-t}\x)-\phi(S^{-t}\y))|<e^{\lambda t} \eps_1,
 \ee
 making it clear that the small enough choice of $\eps_1$ such that $e^{\lambda t} \eps_1<\eps$, and selecting $|\x-\y|<\delta$ such that  $|S^{-t}(\x)-S^{-t}(\y)|<\delta_1$, which is possible by continuity of $S^t$, completes the proof.}
\end{proof}
\begin{cor}

 If $I_\z=(-\infty,\infty), \ \forall \z\in  P$, and $ P$ is a proper subset of $M$, then $\phi$ is a subdomain eigenfunction of $U^t$ on  $ P$ associated with the eigenvalue $\lambda$.
 \end{cor}
\begin{cor} If $I_\z=(-\infty,\infty)$, and $ P=M$  then $\phi$ is an eigenfunction of $U^t$ on  $M$ associated with the eigenvalue $\lambda$.

\end{cor}

\begin{remark}
It is easy to build open eigenfunctions around any point that is not an equilibrium. Namely, in an open set $N(\bp)$ around any non-equilibrium point $\bp$, the flow can be straightened out in coordinates $x_1,...,x_n$ such that $\dot x_1=1,\dot x_j=0, j\neq 1$. Now consider any function $\chi_1(x_2,...,x_n)$ defined on the section $x_1=0$,
and define $\chi (x_1,...x_n)=ce^{\lambda x_1}\chi_1(x_2,...x_n)$
Then,
\be
\dot \chi (x_1,...x_n)=\lambda c e^{\lambda x_1}\chi_1(x_2,...,x_n)\dot x_1=\lambda \chi (x_1,...x_n), 
\ee
for any $\lambda\in\R,$ provided $\chi_1$ is real, and $\lambda\in\C,$ provided $\chi_1$ is complex, and the definition being valid in $N(\bp)$. 
Thus, singularities in state space, such as fixed points, and reccurrencies, such as those occuring in a flow around a limit cycle, serve to select $\lambda$'s in the 
Koopman operator spectrum. \end{remark}
\subsection{Poincar\'e Linearization and Eigenmode Expansion}
 We  consider a continuously differentiable
dynamical system defined in some open region $\cD$ of $\mathbb{R}^n$,
\begin{equation}
\dot{\x}=\F(\x)=A\x+\v(\x)
\,,\label{eq:xdot}
\end{equation}
where the origin $\x=0$ is an equilibrium contained in $\cD$.
The matrix $A=D\F_0$  is the gradient of the vector field at
$\x=0$, and $\v(\x)$ is the ``nonlinear" part of the vector field, $\v(\x)=\F(\x)-A\x$. The system (\ref{eq:xdot}) induces a
flow $S^t(\x): \cD \times \mathbb{R} \to  \cD$ and the positively invariant basin
of attraction\index{Basin
of Attraction} ${\cal B}$ of the fixed point is defined by
\begin{equation}
{\cal B}=\{\x: S^t(\x)\in \cD, \forall t \geq 0, \mbox{ and }
\lim_{t\to \infty}S^t(\x)=0\}. \label{eq:dba}
\end{equation}
The Poincar\'e linearization, valid for analytic vector fields with no resonances \index{Resonance of eigenvalues} amongst eigenvalues, in the neighborhood of a fixed point reads
\begin{defi} Let $\lambda_1, \lambda_2,...,\lambda_n$ be the eigenvalues of $D\F|_0$. We say that $D\F|_0$ is resonant if there are nonnegative integers
$m_1, m_2,...,m_n$ and $s \in \{1, 2,...,n\}$ such that
\be\sum_{k=1}^n
m_k \geq 2
\ee
and
\be\lambda_s = \sum_{k=1}^n
m_k\lambda_k.\ee
If $D\F|_0$ is not resonant, we say that it is nonresonant.
\end{defi}
For analytic vector fields, according to the normal form theory, nonresonance, together with the condition that  all eigenvalues are in the left half plane (stable case) or right half plane (unstable case),  permits us to make  changes of variables that remove
nonlinear terms up to any specified order in the right-hand side of the differential
equation \cite{GuckenheimerandHolmes:2002}. Alternatively, the Siegel condition\index{Siegel condition} is required:
\begin{defi}
We say that $(\lambda_1, \lambda_2,...,\lambda_n) \in \C^n$
 satisfy the Siegel condition if there
are constants $C > 0$ and $\nu > 1$ such that
\be
|\lambda_s - \sum_{k=1}^nm_k\lambda_k|\geq\frac{C}{  (\sum_{k=1}^n
m_k)^\nu
},
\label{eq:sieg}
\ee
 for all nonnegative integers $m_1, m_2,...,m_n$ satisfying
\be
\sum_{k=1}^n
m_k \geq 2.
\ee
\end{defi}
This leads to the Poincar\'e Linearization Theorem:

\begin{ther} [\bf Poincar\'e Linearization Theorem] 
\label{ther:PLT}
Suppose that $\F$ is analytic, $\F(0)=0$, and that all the eigenvalues of  $D\F|_0$ are nonresonant and either all lie in the
open left half-plane, all lie in the open right half-plane, or satisfy the Siegel
condition. Then there is an analytic change of variables $\y=\h(\x)$  such that 
$\dot \y=A\y$ in a small neighborhood $N(0)$ of the fixed point  $0$.
\label{ther:Poinc}
\end{ther}
Poincar\'e linearization is used in normal form theory \cite{Perko:2013,GuckenheimerandHolmes:2002}, and the issue of resonances is the well-known reason that
even analytic vector fields can not always be linearized using an analytic change of variables.

The Koopman group of operators $U^t$ associated with (\ref{eq:xdot}) evolves  a (vector-valued) \emph{observable} $\ve{f}:\mathbb{R}^n \mapsto \mathbb{C}^m$ along the trajectories of the system, and is  defined via the composition 
$$U^t \ve{f}(\ve{x})=\ve{f} \circ S^t(\x).$$
Let $ \phi_j$ be the eigenfunctions of the Koopman group associated with the Poincar\'e linearization matrix $A$. Then $s_j(\x)=\phi_j(\h(\x))$ are the (analytic)  eigenfunctions of the Koopman operator associated with the nonlinear system. Clearly, $s_j(0)=0$. We will utilize 
\be
\y=\s(\x)=(s_1(\x),...,s_n(\x))
\ee as a change of coordinates.

Like in the case of linear systems treated in  section \ref{sect:lin} we'd like to again get an expansion of observable $f$ into eigenfunctions of $U^t$.  If the observable $f$ is analytic, the Taylor expansion of $f(\ve{s}^{-1}(\ve{y}))$  around the origin\footnote{Note that the change of variables in the Poincar\'e-Siegel linearization is an analytic diffeomorphism \cite{Arnold:2012}, and thus $\s^{-1}$ is analytic.} yields
\bea
f(\ve{s}^{-1}(\ve{y}))&=&f(0)+D f^T|_0 \, D{\ve{s}^{-1}} |_0\ve{y} \nonumber \\
&+& \frac{1}{2} \ve{y}^T\, (D{\ve{s}^{-1}}|_0)^T \ve{H} D{\ve{s}^{-1}}|_0 \, \ve{y} + \frac{1}{2} \ve{y}^T \sum_{k=1}^n \left.\frac{\partial f}{\partial x_k}\right|_{0} \ve{H}_{s^{-1}_k} \ve{y}\nonumber \\
& +& \textrm{h.o.t.}\,,
\label{Taylor_expansion}
\eea
where $\ve{H}$ is the Hessian matrix of $f$ at $0$ 
 $$\ve{H}_{ij}=\frac{\partial^2 f}{(\partial x_i \partial x_j)}(0),$$ and $\ve{H}_{s^{-1}_k}$ is the Hessian matrix of $s_k^{-1}$ at the origin 
 $$\ve{H}_{s^{-1}_k,ij}=\frac{\partial^2 s^{-1}_k}{\partial y_i \partial y_j}(0).$$ 
 Using the relationship $\ve{y}=(s_1(\ve{x}),\dots,s_n(\ve{x}))$, we can turn the expansion \eqref{Taylor_expansion} into an expansion of $\f$ onto the products of the eigenfunctions $s_j$. For a vector-valued observable $\ve{f}$, we obtain
\begin{equation}
\label{equa_expansion_f}
\ve{f}(\ve{x})=\sum_{\{k_1,\dots,k_n\}\in\mathbb{N}^n} \ve{\overline{v}}_{k_1\cdots k_n} \, s_1^{k_1}(\ve{x}) \cdots s_n^{k_n}(\ve{x})
\end{equation}
with the  Koopman modes\index{Koopman mode} $\ve{\overline{v}}_{k_1\cdots k_n} $ up to linear terms reading
\begin{equation}
\overline{\ve{v}}_{0,...,0}=\ve{f}(0)   
\ee
\be
\displaystyle
\overline{\ve{v}}_{0,..,1_j,.,0}=\sum_{k=1}^n \left.\frac{\partial \ve{f}}{\partial x_k}\right|_{0} \left.\frac{\partial s_k^{-1}}{\partial y_j}\right|_{0},
\ee
where notation $1_j$ means that there is $1$ at the $j$th place in the sequence. Note that the $(0,...,0)$ Koopman mode is just the time-average of the evolution of $\f$ \cite{Mezic:2005}. We also have 
\be
\displaystyle
\overline{\ve{v}}_{0,..1_i,...,1_j,.,0}=\sum_{k=1}^n \sum_{l=1}^n \left.\frac{\partial^2 \ve{f}}{\partial x_k \partial x_l}\right|_{0} \left.\frac{\partial s_k^{-1}}{\partial y_i}\right|_{0} \left.\frac{\partial s_l^{-1}}{\partial y_j}\right|_{0} + \sum_{k=1}^n \left.\frac{\partial \ve{f}}{\partial x_k}\right|_{0} \left.\frac{\partial^2 s_k^{-1}}{\partial y_i \partial y_j}\right|_{0} \ \  \ee
\be
\displaystyle
\overline{\ve{v}}_{0,..,2_i,...,0}=\frac{1}{2}\sum_{k=1}^n \sum_{l=1}^n \left.\frac{\partial^2 \ve{f}}{\partial x_k \partial x_l}\right|_{0} \left.\frac{\partial s_k^{-1}}{\partial y_i}\right|_{0} \left.\frac{\partial s_l^{-1}}{\partial y_i}\right|_{0} + \frac{1}{2} \sum_{k=1}^n \left.\frac{\partial \ve{f}}{\partial x_k}\right|_{0} \left.\frac{\partial^2 s_k^{-1}}{\partial y_i^2} \right|_{0} 
\end{equation}
The other (higher-order) Koopman modes can be derived similarly from \eqref{Taylor_expansion}. 

For the observable $\f(\ve{x})=\ve{x}$, the Koopman modes are given by
\begin{equation*}
\ve{v}_{k_1\cdots k_n} = \frac{1}{k_1 ! \dots k_n !} \left.\frac{\partial^{k_1\cdots k_n} \ve{s}^{-1}}{\partial^{k_1} y_1 \cdots \partial^{k_n} y_n} \right |_{0} \,.
\end{equation*}
In particular, the eigenvectors of the Jacobian matrix $D\F|_0$ (i.e. $\ve{v}_j=\ve{v}_{k_1\cdots k_n}$, with $k_j=1$, $k_i=0 \, \  \forall i \neq j$) correspond to
\begin{equation*}
\ve{v}_j=\left. \frac{\partial \ve{s}^{-1}}{\partial y_j}\right|_{0}
\end{equation*}
and one has $D{\ve{s}^{-1}}=V$, where the columns of $V$ are the eigenvectors $\ve{v}_j$. 
In addition, the differentiation of $\ve{y}=\ve{s}(\ve{s}^{-1}(\ve{y}))$ at the origin leads to
\begin{equation*}
\delta_{ij} =  D s_i(0)\left.\frac{\partial \ve{s}^{-1}}{\partial y_j}\right|_{0}  =  D s_i(0)  \v_j.
\end{equation*}
Therefore, the gradient $D s_i(0)$ is the left eigenvector $\w_i$ of $D\F|_0$ (associated with the eigenvalue $\lambda_i$) and one has
\begin{equation}
\label{approx_eigen}
s_i(\ve{x}) = D s_i^c|_0\ve{x}+ o(\|\ve{x}\|)= \langle \ve{x},\w_i \rangle +o(\|\ve{x}\|) \,,
\end{equation}
which implies that, for $\|\ve{x}\|\ll 1$, the eigenfunction $s_i(\ve{x})$ is well approximated by the eigenfunction of the linearized system.

From (\ref{equa_expansion_f})  the spectral decomposition for evolution of a vector-valued  analytic observable $\f$ is given in $N(0)$, for $t\in\R^+$ by
\begin{equation}
\label{evol_observable}
U^t \ve{f}(\ve{x}) = \sum_{\{k_1,\dots,k_n\}\in\mathbb{N}^n} s_1^{k_1}(\ve{x}) \cdots s_n^{k_n}(\ve{x}) \, \ve{\overline{v}}_{k_1\cdots k_n} \, e^{(k_1 \lambda_1+\cdots+k_n \lambda_n) t}\,
\end{equation}
and the vectors $\ve{\overline{v}}_{k_1\cdots k_n}$ are the Koopman modes, \index{Koopman mode} i.e. the projections of the observable $\ve{f}$ onto $s_1^{k_1}(\ve{x}) \cdots s_n^{k_n}(\ve{x})$. For the particular observable $\ve{f}(\ve{x})=\ve{x}$, \eqref{evol_observable} corresponds to the expression of the flow and can be rewritten as
\begin{equation}
\label{sol_nonlin}
 U^t \ve{x} =\sum_{j=1}^n s_j(\ve{x}) \ve{v}_j \, e^{\lambda_j t} + \sum_{\substack{\{k_1,\dots,k_n\}\in\mathbb{N}_0^n\\ k_1+\cdots+k_n>1}} s_1^{k_1}(\ve{x}) \cdots s_n^{k_n}(\ve{x}) \,\ve{v}_{k_1\cdots k_n} \, e^{(k_1 \lambda_1+\cdots+k_n \lambda_n) t}\,.
\end{equation}
The first part of the expansion is similar to the linear flow (\ref{eq:x_expansion}). We can use these results to show that the eigenvalues $\lambda_j$ and the Koopman modes $\ve{v}_j$ are the eigenvalues and eigenvectors of $D\F|_0$. 
The vectors $\ve{\overline{v}}_{k_1\cdots k_n}$ are the so-called Koopman modes \cite{Rowley}, i.e. the projections of the observable $\ve{f}$ onto $s_1^{k_1}(\ve{x}) \cdots s_n^{k_n}(\ve{x})$. 

For an equilibrium at $\x^*$ instead of at $0$, and for the particular observable $\ve{f}(\ve{x})=\ve{x}$, we get \bea
\label{sol_nonlin}
& &U^t \ve{x} =\ve{x}^* + \sum_{j=1}^n s_j(\ve{x}) \ve{v}_j \, e^{\lambda_j t} 
+\!\!\!\!\!\!\sum_{\substack{\{k_1,\dots,k_n\}\in\mathbb{N}_0^n\\ k_1+\cdots+k_n>1}} \!\!\!\!\!\!\!\!\!\!\!\!s_1^{k_1}(\ve{x}) \cdots s_n^{k_n}(\ve{x})
 \cdot \ve{v}_{k_1\cdots k_n}e^{(k_1 \lambda_1+\cdots+k_n \lambda_n) t} 
\eea
where $\x^*$ is the time average of the state. This will be the term that also comes out in the case of the more general attractors treated below.

{\color{black} Provided the Taylor expansion (\ref{Taylor_expansion}) has validity in all of ${\cal B}(0)$}, we can utilize Lemma \ref{lem:ext} to extend the validity of eigenfunctions from $N(0)$ to the whole basin of attraction of $0$:
\begin{proposition} Let ${\cal B}(0)$ be the basin of attraction of $0$. The spectral expansion of a vector-valued  analytic observable $\f$ in ${\cal B}(0)$ is given by
\begin{equation}
\label{evol_observable1}
U^t \ve{f}(\ve{x}) = \sum_{\{k_1,\dots,k_n\}\in\mathbb{N}^n} \tilde s_1^{k_1}(\ve{x}) \cdots \tilde s_n^{k_n}(\ve{x}) \, \ve{\overline{v}}_{k_1\cdots k_n} \, e^{(k_1 \lambda_1+\cdots+k_n \lambda_n) t}\,
\end{equation}
where $\tilde s_j$ are (possibly subdomain, if  ${\cal B}(0)\neq M$) eigenfunctions of $U^t$.
\end{proposition}
\begin{proof}
We obtain $\tilde s_j$ from $s_j$ defined on the open set $N(0)$ by pulling back by the flow, as in Lemma \ref{lem:ext}. Namely, ${\cal B}(0)=\cup_{t\in\R^+}S^{-t}N(0)$. Since $\tilde s_j=s_j$ on $N(0)$, using equation (\ref{evol_observable}) that is valid in $N(0)$, we see that the statement (\ref{evol_observable1}) is true for $t\in\R^+$. Let $\z\notin N(0)$. We have \bea
\label{evol_observable2}
U^\tau \ve{f}(\ve{z}) &=&U^\tau \ve{f}(S^{-t(\z)}\x) 
= \sum_{  \{k_1,\dots,k_n\}\in\mathbb{N}^n }
= s_1^{k_1}(\ve{x}) \cdots   s_n^{k_n}(\ve{x}) 
  \,  \ve{\tilde{v}}_{k_1\cdots k_n} \, 
  e^{(k_1 \lambda_1+\cdots+k_n \lambda_n) \tau},
\eea
where $ \ve{\tilde{v}}_{k_1\cdots k_n}$
 are the Koopman modes associated with $\f\circ S^{-t(\z)}$. Now we need the following lemma:
\begin{lemma} The Koopman modes of $\f$ and $\f\circ S^{-t(\z)}$ are related by
\be 
\ve{\tilde{v}}_{k_1\cdots k_n}=e^{-(k_1 \lambda_1+\cdots+k_n \lambda_n) t(\z)}\ \ve{\overline{v}}_{k_1\cdots k_n}.
\label{eq:KMext}
\ee
\end{lemma}
\begin{proof} This is a simple consequence of the Generalized Laplace Analysis theorem \cite{MohrandMezic:2014}.\end{proof}
\noindent Now we finish the proof of the proposition: Since $$\tilde s_j(\z)=e^{-\lambda_j t(\z)}s_j(S^{t(\z)}\z),$$ combining (\ref{eq:KMext}) and (\ref{evol_observable2}) we obtain (\ref{evol_observable1}).
\end{proof}

We recognize here that the operator formalism we are developing leads to a striking realization: the only difference in the representation of the dynamics of linear and nonlinear systems with equilibria on state space is that  in the linear case  the expansion is finite, while in the nonlinear case it is infinite. In linear systems, we are expanding the state $\x(\bp)$ (which itself is a linear function of a point $\bp$ on the Euclidean state-space), in terms of eigenfunctions of the Koopman operator that are also linear in state $\x(\bp)$. In the nonlinear case, this changes - the Koopman eigenfunctions are in general {\it nonlinear} as functions of state $\x(\bp)$ and the expansion is infinite. It is also useful to observe that the expansion is {\it asymptotic} in nature - namely, there are terms that describe evolution close to an equilibrium point, and terms that have higher expansion or decay rates. 
\subsection{Hartman-Grobman Type Theorems and Stable and Unstable Manifolds}

The benefit of non-resonance conditions or Siegel condition is that the conjugacy $\h$ is analytic. Thus, if we want to expand an analytic vector of observables $\f(\x)$, in terms that reflect dynamics of $\dot \x=\F(\x)$ that has a globally stable equilibrium,  that expansion is readily available by Taylor expanding $\f\circ \h^{-1}(\y)$. 

Theorems of Hartman and Grobman \index{Hartman Grobman Theorem} require much less smoothness, and do not have resonance conditions associated with them. We state Hartman's version, modified slightly to fit into our narrative of Koopman operator theory:

\begin{theorem}[Hartman]
\label{ther:hart}
Let 
\be 
\dot \x=\F(\x), \ \F\in C^2(\cD)\ ,  \F(0)=0,
\label{eq:nlsys}
\ee and $U^t$ the associated Koopman family of operators.
If all of the eigenvalues of the matrix $A=D\F|_0$
have non-zero real part, then there exists a $C^1$-diffeomorphism
 $\tilde \h$ of a neighborhood ${\cal N} $ of $\x=0$ onto an open set ${\cal V}$ containing the origin
such that for each $\x \in {\cal N}$ there is an open interval $I(\x)\subset \mathbb{R}$
containing zero such that for all $\x \in {\cal N}$ and $t \in I(\x)$
\be
U^t \tilde \h(\x)=\tilde \h\circ S^t(\x)=e^{At}\tilde \h(\x).
\label{Hart}
\ee
\end{theorem}
The time interval $I(\x)=(a(\x),b(\x))$ can be extended  to $I(\x)=(a(\x),+\infty)$ provided all eigenvalues of $D\F|_0$ have negative real parts, and to $I(\x)=(-\infty,b(\x))$ provided all eigenvalues of $D\F|_0$ have positive real parts.

Looking at equation (\ref{Hart}), we could call a matrix $A$ an eigenmatrix\index{Eigenmatrix} of $U^t$ associated with eigenmapping\index{Eigenmapping} $\tilde \h$. Within the Hartman theorem, this is the case only {\it locally}, around an equilibrium point, and possibly for finite time, if the equilibrium point is a saddle.
The Hartman theorem therefore states that, locally, the nonlinear system  $\dot \x=\F(\x),$ is conjugate to a linear system $\dot \y=A\y,$ where $\y=\tilde \h(\x)$.

Now, assume $A$ have distinct real eigenvalues. Then, it can be transformed into a diagonal matrix $\Lambda$ using a linear transformation $V$. Setting $\z=V^{-1}\y$
leads to
\be
\dot \z=V^{-1}AV\z=\Lambda \z.
\label{eq:Hartdiag}
\ee
Using $V^{-1}$, from (\ref{Hart}) we also get
$$
V^{-1}\tilde \h\circ S^t(\x_0)=V^{-1}e^{At}\tilde \h(\x_0)=V^{-1}e^{At}VV^{-1}\tilde \h(\x_0).
$$
Thus, $\tilde \k=V^{-1} \tilde \h$ satisfies

\begin{equation}
 \tilde \k\circ S^t(\x_0)=e^{\Lambda t} \tilde \k(\x_0)
\,
\label{Koop1}
\end{equation}
i.e. each component function of $\k$ is an eigenfunction of $U^t$. Thus, we proved, that $\dot \x=\F(\x)$ is conjugate to the diagonal linear system   (\ref{eq:Hartdiag})  in ${\cal N}$, and the conjugacy is provided by the mapping $\tilde \k$ whose components are Koopman eigenfunctions. 

Hartman's local theorem for stable
equilibria can be extended to a global one that is valid in the whole basin of
attraction ${\cal B},$ as shown in \cite{LanandMezic:2013}:

\begin{theorem}[Autonomous flow linearization]
\label{ther:pullback}
Consider the system (\ref{eq:xdot}) with
$\v(\x)\in C^2(\cD)$. Assume that $A$ is a $n\times n$ Hurwitz
matrix, {\em i.e.} all its eigenvalues have negative real parts (thus $\x=0$
is exponentially stable). Let ${\cal B}$ be the basin of attraction of $0$. Then
$\exists\, \h(\x)\in C^1({\cal B}):{\cal B} \to \mathbb{R}^n$, such that
$\y=\h(\x)$ is a $C^1$ diffeomorphism with $\mathbf{D}\h(0)=I$ in ${\cal B}$
and satisfies $\dot{\y}=A\y$.
\label{ther:stb}
\end{theorem}
\begin{proof}
The proof is based on the following observation: Hartman theorem provides us with a domain $U$ inside which
the local conjugation $\tilde \h$  to the linear system $\x=A\x$,  exists. If we find a manifold diffeomorphic  to a
sphere $\Sigma$ of dimension $n-1$ inside $U$ such that each initial point $\x\in {\cal B}$ has a unique point $i(\x)$ (and unique time, $t(\x)$) of intersection with $\Sigma$,
then the mapping
$$
\h(\x)=e^{-At(\x)}\tilde \h(S^{t(\x)}(\x))
$$
is the required conjugacy. To see this, observe that 
\begin{eqnarray}
\h(S^\tau(\x))&=&e^{-A(t(\x)-\tau)}\tilde \h(S^{t(\x)-\tau}(S^\tau(\x)))\nonumber \\
&=&e^{A\tau}e^{-At(\x)}\tilde \h(S^{t(\x)}(\x))\nonumber \\
&=&e^{A\tau}\h(\x),\nonumber \\
\label{eq:conjfp}
\eea
where we used the fact that $t(S^\tau(\x))=t(\x)-\tau$.
Since the surface $\Sigma$ exist by the  converse Lyapunov theorem \cite{Vidyasagar:2002},\footnote{Specifically, if an equilibrium is asymptotically stable 
from an open set $U$, then there is a Lyapunov function $L$ such that, sufficiently close to the origin (but not at the origin) $\dot L<0$ and thus the vector field ``points inwards" on level sets of $L$ sufficiently close to the origin. We can choose $\Sigma$ to be one of those level sets. Then clearly $t(\x)$ and $i(\x)$ are unique, for every trajectory as if not, the trajectory would need to ``enter" and then ``exit" the interior of $L$.} the theorem is proven.
\end{proof}
The following corollary, that enables extension of eigenfunctions to the whole basin of attraction holds:
\begin{corollary} The functions $\k=V^{-1}\h$ are eigenfunctions of (\ref{eq:nlsys}) in the basin of attraction ${\cal B}$ of $0$.
\end{corollary}
When the equilibrium is a saddle point, the result on extension of Koopman eigenfunctions can be obtained using Lemma \ref{lem:ext}:
\begin{proposition} Let the equilibrium $0$ of (\ref{eq:nlsys}) be a non-degenerate saddle point, i.e. all the eigenvalues of $A$ have non-zero real values. Then, the $n$ functions $\k=V^{-1}\h$ can be extended to open eigenfunctions of the Koopman operator associated with (\ref{eq:nlsys}) in the set
$ P$ defined in Lemma  \ref{lem:ext}, by setting
\be
\s(\z)=e^{-\lambda t(\z)}\k(S^{t(\z)}\z),\ \forall \z\in  P.
\ee
\end{proposition}
The following definition characterizes the principal parts of stable and unstable manifolds of the equilibrium point in $P$:
\begin{definition} Let $W^{s(u)}_{ P}$ be the part of the stable (unstable) manifold of $0$ such that for every $\z\in W^{s(u)}_{ P}$ we have
$S^{t(\z)}\z\in W^{s(u)}_{loc},$ and $W^{s(u)}_{loc}$ is the local stable (unstable) manifold at $0$. \end{definition}
We have the following corollary:
\begin{corollary}
Let $s_1,...,s_u$ be open eigenfunctions of the Koopman operator on $ P$ associated with the positive real part eigenvalues, and 
let $s_{u+1},...,s_{n}$ be open eigenfunctions of the Koopman operator on $ P$ associated with the negative real part eigenvalues.
Then the joint level set of (generalized) eigenfunctions
\be L_s=\{\x\in\R^n|s_1(\x)=0,...,s_{u+c}(\x)=0\},\ee is  $W^{s}_{ P}$, 
 and 
\be L_u=\{\x\in\R^n|s_{u+1}(\x)=0,...,s_n(\x)=0\},\ee is    $W^{u}_{ P}$.
\end{corollary}
\subsection{Center Manifolds}
\label{nonlman}
Now we tackle the problem of defining the (global) center manifold for a nonlinear systems using Koopman operator eigenfunctions.
 Let $0$ again be an equilibrium point of a smooth nonlinear system, 
\be
 \dot \z=\F(\z), \ \z\in \R^{n},
 \label{eq:nonl}
\ee
 with eigenvalues $\lambda_j, j=1,..,s$  associated with the  linearization $D\F|_0$ at equilibrium. Let $s,c,u$ be the number of negative real part eigenvalues, $0$ and positive real part eigenvalues of $D\F|_0$. Let  $(\lambda_1,...\lambda_u)$ be positive real part eigenvalues, $(\lambda_{u+1},...,\lambda_{u+c})$ $0$ real part eigenvalues, and $(\lambda_{u+c+1},...,\lambda_{u+c+s})$ be negative real part eigenvalues of $D\F|_0$.
 
 We split $D\F|_0$ into the zero real part eigenvalue block $B$ - a $c\times c$ matrix - and the non-zero real part eigenvalue block $A$ - an $(u+s)\times (u+s)$ matrix. The  Palmer linearization theorem \cite{KirchgraberandPalmer:1990} generalizes the Hartman-Grobman theorem in this situation:
\begin{ther}\label{KP}
Let the equation (\ref{eq:nonl}) be written as
\bea
\dot \x&=&B\x+g(\x,\y), \label{eq:form01} \\
\dot \y&=&A\y+h(\x,\y), 
\label{eq:form02}
\eea
where $h,g$ are bounded and  Lipshitz with sufficiently small Lipshitz constants $c_j,d_j,j=1,2$:
\bea
|g(\x_1,\y_1)-g(\x_2,\y_2)|&\leq& c_1|\x_1-\x_2|+c_2|\y_1-\y_2|, \\
|h(\x_1,\y_1)-h(\x_2,\y_2)|&\leq& d_1|\x_1-\x_2|+d_2|\y_1-\y_2|,
\eea
Then, (\ref{eq:form01}-\ref{eq:form02}) is $C^0-$ conjugate to the system  
\bea
\dot{\tilde\x}&=&B\tilde\x+\tilde g(\tilde\x),\label{eq:form0} \\
\dot{\tilde\y}&=&A \tilde\y, 
\label{eq:form1}
\eea
where $\tilde g$ is a bounded, Lipshitz function.
\end{ther} 
\begin{prop} \label{prop:center}The system (\ref{eq:form01}-\ref{eq:form02}) has $u$ unstable (generalized)  eigenfunctions $s_1,...,s_u$ and $s$ stable (generalized) eigenfunctions $s_{u+1},...,s_{u+s}$ of the Koopman operator.
The joint zero level set of unstable (generalized) eigenfunctions
\be L_{cs}=\{\x\in\R^n| s_1(\x)=0,...,s_{u}(\x)=0\},\ee is the center-stable manifold $W^s$, 
\be L_c=\{\x\in\R^n|s_1(\x)=0,...,s_{u}(\x)=0,s_{u+1}(\x)=0...,s_{u+s}(\x)=0\},\ee is the (global, unique) center manifold $W^c$, and 
\be L_{cu}=\{\x\in\R^n| s_{u+1}(\x)=0,...,\phi_{u+s}(\x)=0\},\ee is the center-unstable subspace $W^{cu}$.
\end{prop}
\begin{proof} 
Theorem \ref{KP} provides us with functions $\tilde \y(\x,\y)$ that satisfy 
\be\dot{\tilde\y}=A \tilde\y.\ee Let $V$ be the matrix  transforming $A$ to its Jordan normal form, $J$.
Define $\s=V^{-1}\tilde \y$. Then, 
\be
\s=(\s_u,\s_s)=(s_1,...,s_u,s_{u+1},...,s_{u+s}),
\ee
 satisfy equation (\ref{eq:jord}), i.e.
\be
\dot{\s} =J\s,
\ee
and thus, by Corollary \ref{cor:genef}, and Proposition \ref{prop:conj}, $\s$ is a set of generalized eigenfunctions of the Koopman operator associated with the system. The joint zero level set of these is the global center manifold on which the dynamics is given by 
\be
\dot{\tilde\x}=B\tilde\x+\tilde g(\tilde\x).
\ee
Consider now the set
\be L_{cs}=\{\x\in\R^n| s_1(\x)=0,...,s_{u}(\x)=0\},\ee
the dynamics on which is given by 
\bea
\dot{\tilde\x}&=&B\tilde\x+\tilde g(\tilde\x), \nonumber \\
\dot \s_s&=&J_s \s_s, 
\label{eq:form2}
\eea
where $J_s$ is the Jordan form block corresponding to eigenvalues with negative real part. This proves that $L_{cs}=W^{cs}$,
the center-stable manifold of $0$. The proof for the center-unstable manifold, $W^{cu}$ is analogous. \end{proof} 
\begin{rem}
The requirement of the boundedness of nonlinear terms in the above result is not necessarily an obstacle to defining the global center manifold. Namely, it is often stated that only the local center manifold can be obtained, since a bump-function type correction to the vector field needs to be introduced to control the possibly non-Lipshitzian growth away from the small neighborhood of $0$. However, this is not necessarily so. Consider again the system (\ref{eq:xdot}). Assume now there is a factor-conjugacy $g:{\cal D}\subset \R^n\rar {\cal C}\subset\R^1$ such that 
\be
 \exp(\lambda t)g(\x)=g(S^t(\x)).
 \label{eq:scFen}
 \ee
 Then clearly we have found a Koopman eigenfunction for the fully non-linear system. But such conjugacies can exist if the stable and unstable manifolds for the nonlinear case come with associated fibrations and projections to stable and unstable directions that commute with the flow. Consider for example fibration of the stable manifold $F^s(W^s)$ and the associated projection $\Pi^s:F^s(W^s)\rar W^s$ that satisfies 
 $$
 S^t(\Pi^s(\x))=\Pi^s(S^t(\x)).
 $$
  Since there is a conjugacy mapping $\h:W^s\rar E^s$, where $E^s$ is the stable subspace \cite{LanandMezic:2013}, we set $$\g(\x)=\h\circ \Pi^s(\x).$$ Then,
 \be
 \g(S^t\x)=\h\circ \Pi^s(S^t\x)=e^{A_s t}(\h(\Pi^s(\x)))=e^{A_s t}\g(\x),
 \label{eq:sc1}
 \ee
 where $A_s$ is the stable block of the matrix $A$, and we have achieved the desired semi-conjugacy. Multiplying  (\ref{eq:sc1}) from the left by matrix $V^{-1}$, where the matrix $V$ takes $A_s$ into its Jordan form, we obtain $s$ stable (generalized) eigenfunctions of the Koopman operator with the same eigenvalues as those of $A_s$. In the same way we can construct $u$ unstable eigenfunctions. We then define the (global) center manifold as the joint zero level set of these stable and unstable eigenfunctions. The center-stable and the center-unstable manifolds can now be defined analogously to the 
 definition in the Proposition \ref{prop:center}.
 
 Consider the famous example (attributed to Kelley \cite{Kelley:1967}) \footnote{Kelley in fact presented a similar example, with a stable direction instead of the unstable one in $x_2$.} of a vector field with non-unique local center manifold,
\be\left( \begin{array}{c} \dot x_1 \\ \dot x_2  \\ \end{array} \right)= \left( \begin{array}{c}
(x_1)^2 \\
x_2 \\
\end{array} \right), 
\label{eq:CenNonUn}
\ee
The figure depicting the flow is \ref{fig:CenNonUn}
\begin{figure}[h!!] 
\centering \hspace{.5cm} \includegraphics[width=11.5cm,
height=6.5cm, clip=true, trim=0 100 0 100]{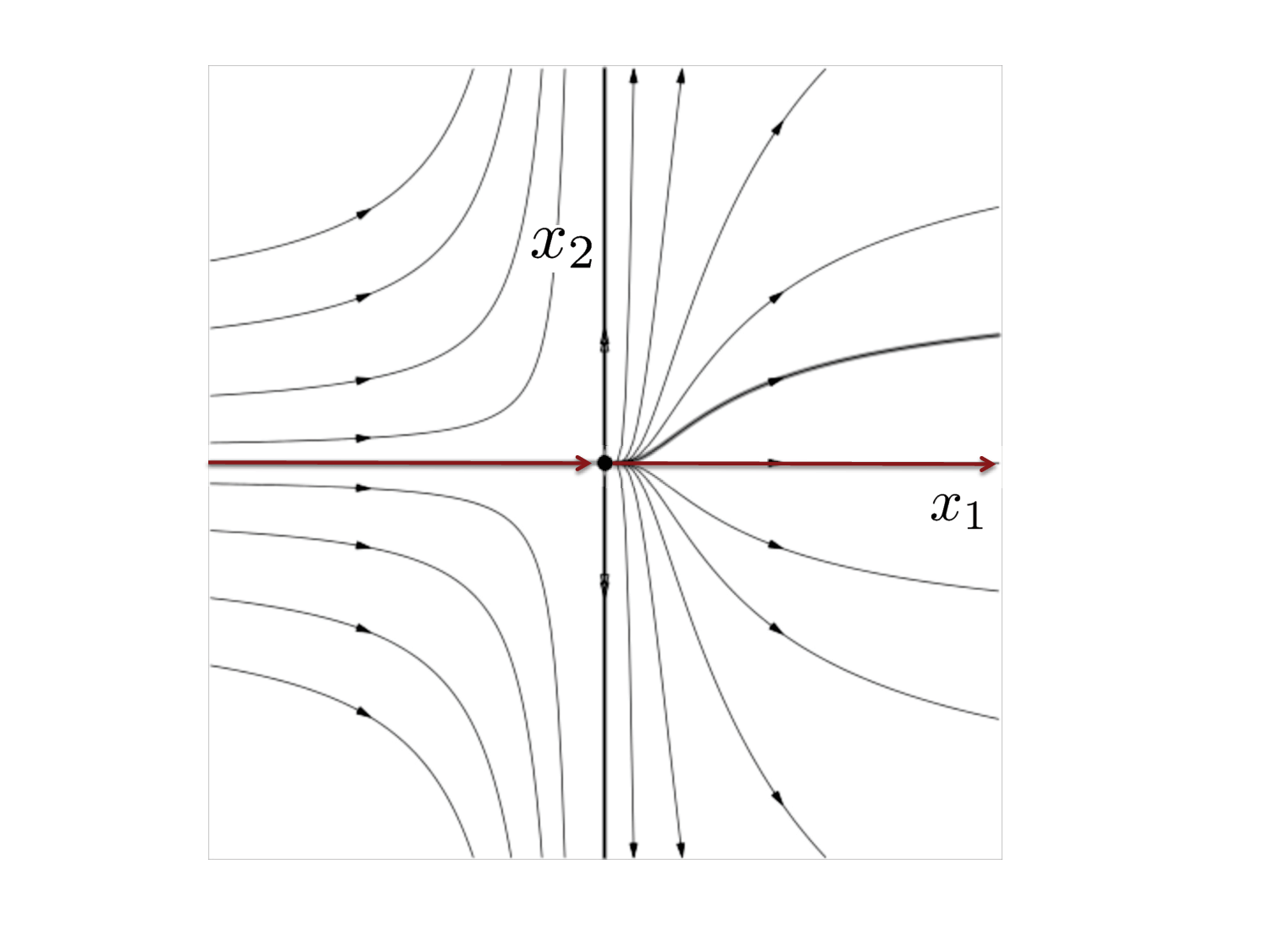}
\caption{Flow of vector field (\ref{eq:CenNonUn}). The ``true" global center manifold is shown in red. }
\label{fig:CenNonUn}
\end{figure}

  In the Kelley example above, although all the curves 
  \be 
  x_2=Ae^{-1/x_1},x_1>0,
  \ee (one of which is shown in figure  \ref{fig:CenNonUn} in bold black)  joined with $x_2=0$ axis for $x_1\leq 0$ satisfy the requirement of tangency of local center manifold to center subspace, they have exponential behavior as $t\rar -\infty$, and in the operator-theoretic point of view fail to satisfy the global center manifold properties. In contrast, the construction we provided above would indicate the global center manifold in this case is $x_2=0$, i.e. the $x_1$ axis, corresponding to the zero level set of the eigenfunction $f_2(x_1,x_2)=x_2$. Note also that, 
\be
f_1=\left\{\begin{array}{c  c }
  e^{1/x_1}&\ \forall x_1< 0  \\
  0 &\ \forall x_1\geq0
  \end{array} \right.
 \ee  satisfies, for $x_1<0$
 \be
  \dot f_1=-\frac{1}{x_1^2}e^{1/x_1}\dot x_1=-e^{1/x_1}=-f,
\ee
  For $x_1\geq 0$ the equation $\dot f=-f$ is trivially satisfied, and thus $f$ is an (infinitely smooth, but not analytic!) eigenfunction of the system at eigenvalue $-1$. It is interesting to note that $x_1\geq 0$ (the zero level set of $f$)  is the unstable manifold (with boundary) of $0$, despite the fact that $(f(x_1),x_2)$ does not
  represent a new coordinate system on the plane as it does not distinguish points with the same $x_2$ on the right half plane.
  
Note that  there is another smooth eigenfunction, \be
f_3=\left\{\begin{array}{c  c }
  e^{-1/x_1}&\ \forall x_1\geq 0 \\
  0 &\ \forall x_1<0
  \end{array} \right.
 \ee 
 corresponding to eigenvalue $1$.  In addition, the product of two eigenfunctions, $f_1$ and $f_2$, $f=x_2e^{1/x_1}$ is another eigenfunction whose level sets are trajectories for $x_1\geq 0$. The intersection of zero level sets of two unstable eigenfunctions, $f_2=x_2=0$ and $f_3$ (the whole semiline $x_1\leq0$) is still the stable manifold of $0$.
  \end{rem}
  {\color{black}
\section{Function Spaces for Dissipative Dynamical Systems}
\label{sect:hilb}
The next couple of sections are dedicated to deduction of spectral expansions of dissipative systems with limit cycling and quasi-periodic attractors. To prepare for this, in this section we consider the general problem of finding function spaces in which spectral expansion is possible. For this task, we do not need to restrict to a specific attractor type. The definition of a dissipative system we use here is the one that has a global (Milnor) attractor $\caA$ of zero Lebesgue measure,  where $\caA$ is equipped with a physical measure \cite{climenhagaetal:2017}.

Let $\cD$ be a compact forward invariant set of $\S^t$, and  $\mu$ the physical invariant measure on the (Milnor) attractor $\caA \subset \cD$, with $\mu(\caA)=1$. Let $\cH_\caA$ be the space of square-integrable observables on $\caA$. Note that $U^t$ restricted to $\cH_\caA$ is unitary. The functions in $\cH_\caA$ can typically be thought of as being defined on the whole basin of attraction of  $\caA$. 
 In a paper by Sell \cite{Sell:1983} it is shown that in the vicinity of a smooth compact $m$-dimensional invariant attracting manifold $\caA$ coordinates $(\u,\y)$ can be found such that the vector field can be written as
\bea
\dot \y=A(\u)\y+\v(\y,\u) \nonumber \\
\dot\u=\bomega(\u)+\bOmega(\y,\u),
\label{eq:gen1}
\eea
where $\v(0,\u)=D_{\y}\v(0,\u)=\bOmega(0,\u)=0$ and thus on $\caA$ the dynamics is given by $\dot\u=\bomega(\u)$. One can think of $\u$ as ``on-attractor" coordinates and $\y$ as ``off-attractor" coordinates. This coordinatization can be extended in certain cases to the full basin of attraction using the following pullback construction:  it was shown in \cite{Sell:1983} that,  in the case when $\caA$ is normally hyperbolic and satisfies certain spectral conditions, (\ref{eq:gen1}) is $C^s$ conjugate in a neighborhood of $\caA$ to the skew-linear\index{Skew-linear dynamical system} system 
\bea
\dot \y=A(\u)\y \nonumber \\
\dot\u=\bomega(\u),
\label{eq:genlin}
\eea
by the conjugacy $\tilde \h$. We call such an attractor $\caA$ the Sell-type attractor.
Let $S_\caA^t(\u)$ be the flow restricted to the manifold $\caA$, $S_{\caA\times \caA^T}^t(\y,\u)$ the flow of (\ref{eq:genlin}) and $S^t(\x)$ the flow of the original system 
\be
\dot \x=\F(\x).
\label{Eq:orig}
\ee
Close to the attractor $\caA$ we can pick a section $\Sigma$ of the normal bundle\index{Normal bundle} 
such that every trajectory of (\ref{Eq:orig}) eventually crosses $\Sigma$, and does the crossing only once. The global conjugacy $\h$ in the basin of attraction ${\cal B}$ of $\caA$ is given by 
$$
\h(\x)=S_{\caA\times \caA^T}^{-t_\Sigma}\tilde\h(S^{t_\Sigma}(\x))
$$
Thus, every point in the basin of attraction of $\caA$ is coordinatized by $(\u,\y)$. In view of this, with a slight abuse of notation, we will consider $\cH_\caA$ as a space of functions defined on all of $\cD$.
\begin{remark}
Provided that $U^t_{S_\caA}$, the Koopman operator associated with the flow on $\caA$, has $m$ smooth eigenfunctions $$(s_1(\u),...,s_m(\u)),$$ when viewed as functions of $\x$ these are eigenfunctions of the Koopman operator of the full system (\ref{Eq:orig}), 
$\s=(s_1(\u(\x)),...,s_m(\u(\x)))$. Define the set ${\cal I}_{\s}$ in the basin of attraction of $\caA$ by
$$
{\cal I}_\s=\{\x\in {\cal B}|\lim_{t\rar\infty}|S^t(\x)-S_\caA^t(\s)|=0.\}
$$
The joint level sets of eigenfunctions $\s$ are sets ${\cal I}_\s$. We call them the $\s$-Isochrons\index{Isochrons, $\s$-isochrons}, in analogy with isochrons for limit cycles and invariant tori \cite{winfree:1974,guckenheimer:1975,MauroyandMezic:2012}.
\end{remark}
 Let $\tilde\cH_\cB\subset C(\cD)$ be a Hilbert space of  functions $f:\cD\rightarrow \mathbb{C}$ orthogonal 
to $\cH_\caA$ with respect to $\mu$, that satisfy
\begin{equation}
\int_\cD f\phi d\mu=0, \ \forall \phi \in \cH_\caA.
\label{eq:ort}
\end{equation}

\begin{remark} The requirement (\ref{eq:ort}) can be replaced by the requirement that $f$ is $0$ on $\caA$ provided $\mu$ is a finite Borel measure, by the following argument of Professor Dimitris Giannakis:
By  (\ref{eq:ort}), every function $f$ in $\tilde\cH_\cB$ lies in the orthogonal complement of $\cH_\caA$ in itself, i.e., the $0$ subspace of $\cH_\caA$. As a result, every $f$ in $\cH_\cB$ is zero $\mu$-a.e., i.e., there exists a measurable subset $\cal{S}$ of $\cD$ with $\mu({\cal S})=1$, such that $f$ vanishes everywhere on $\mu(\cal{S})$. Because $\mu$ is a finite Borel measure, every such $\cal{S}$ must  necessarily be dense in the support of $\mu$. Thus for every $\x$ in the support of $\mu$ there exists a sequence $\{\x_n\}$ in $\cal{S}$ such that $\lim_{n\rar\infty} \x_n = \x$, and because $f$ is continuous, $f(\x) = \lim f(\x_n) = 0$.
\end{remark}
We assume that the tensor product space  $\tilde\cH=\cH_\caA\otimes \tilde\cH_\cB$ is invariant under $U^t$.  It is clear that the space $\cH_\caA$ is itself invariant under $U^t$. So is $\tilde\cH_\cB$:
\begin{lemma} $\tilde\cH_\cB$ is invariant under $U^t$.
\end{lemma}
\begin{proof} We have
\begin{equation}
\int_\cD U^t f(\x)\phi(\x) d\mu(\x)=\int_\cD  f(\x)\phi (\S^{-t}\x)d\mu(\S^{-t}\x)=\int_\cD  f(\x)\psi (\x)d\mu(\x)=0, \ \forall \phi \in \cH_\caA,
\end{equation}
where the last equation is due to the invariance of $\mu$. Since the set of functions $\psi=\phi (\S^{-t})$, is precisely $\cH_\caA$, this
implies that $U^t f \in \tilde \cH_\cB$.
\end{proof}
Define $\cH_\cB=\tilde \cH_\cB\cup {\mathbf 1}$, where ${\mathbf 1}$ is the constant unit function on $\cD$. Recall that a bounded operator $T$ on a space $\cL$ is of scalar type  if it satisfies \
\begin{equation}
T=\int_\C  \lambda dP_{\lambda},
\end{equation}
where $dP$ is the resolution of the identity for $S$ \cite{dundford:1954}.
\begin{rem} The results on linear systems in section \ref{sec:lingen} were obtained in the case when the underlying operator is {\it spectral}, defined as 
$$
T=S+N,
$$
where $S$ is of scalar type, and $N$ is quasi-nilpotent, i.e.
$$
\lim_{n \rar \infty}(T^n)^{1/n}=0.
$$ For simplicity here we assume that the quasinilpotent part is $0$.
\end{rem}
We assume  that $U^t$ restricted to $\cH_\cB$ is a scalar type operator.
Define the tensor product space \begin{equation}\cH=\cH_\caA\otimes \cH_\cB.\end{equation} Clearly,
$U^t=U^t|_{\cH_\caA}\otimes U^t|_{\cH_\cB}$ on $\cH$. Define  $P(a,b)=a\cdot b, \ a,b\in \bC$ to be the scalar product of $a$ and $b$, and 
\begin{equation}
P(A,B)=\cup_{a\in A,b\in B}P(a,b), \ \ A,B\subset \bC.
\end{equation}
We have the following:
\begin{theorem} \label{dissHilb}Consider the composition operator $U^t:\cH\rightarrow \cH$, and let $\sigma(U^t|_{\cH_\caA}),\sigma(\koop|_{\cH_\cB})$ be the spectra of its restrictions to $\cH_\caA$ and $\cH_\cB$ with the associated projection-valued spectral measures $P_\omega, \ \omega\in S^1$, and $P_z,\ z\in \mathbb{C}$. Then $\sigma(U^t)=\cl({P(\sigma(U^t|_{\cH_\caA}),\sigma(U^t|_{\cH_\cB})))}$ and 
\label{the:product}
\begin{equation}
U^t=\int_{\mathbb{C}} \int_{\mathbb{R}} e^{z t} e^{i2\pi\omega t} dP_{\omega}dP_{z}.
\label{spectral}
\end{equation}
\end{theorem}
\begin{proof}
Using  \cite{ReedandSimon:1972}, and the fact that $\cH_\caA,\cH_\cB$ are invariant under $U^t$, it follows that \begin{equation}\sigma(U^t)=\cl({P(\sigma(U^t|_{\cH_\caA}),\sigma(U^t|_{\cH_\cB}))}).
\end{equation}
Since $U^t|_{\cH_\caA}$ is unitary, it satisfies
\begin{equation}
U^t|_{\cH_\caA}=\int_{\mathbb{R}}  e^{i2\pi\omega} dP_{\omega},
\end{equation}
while the composition operator $U^t|_{\cH_\cB}$ is bounded (as $\cH_\cB\subset C$ and $\cD$ is compact)  and thus spectral, and satisfies 
\begin{equation}
U^t|_{\cH_\cB}=\int_\C  e^{z t} dP_{z}.
\end{equation}
Taking into account the product space structure, we obtain (\ref{spectral}).
\end{proof}

\begin{example}
Consider a one-dimensional system 
\be
\dot x=F(x),\ x\in \R,
\label{eq:one}
\ee
 $F$ entire, with a globally attracting, nondegenerate fixed point at $0$, for which $df/dx(0)<0$. The invariant measure is the Dirac delta at $0$. The space $\cH_\caA$ is the space of all constant functions on $\R$, for which the inner product is the usual scalar product. Let $\tilde \cH_\cB$ be the space of analytic (entire) functions 
\begin{equation}
f(z)=\sum_{n=1}^{\infty} a_n z^n, 
\label{eq:Tay}
\end{equation}
where $z\in \C$, and 
\begin{equation}
\int_\C |f(z)|^2 e^{-|z|^2}dz<\infty
\end{equation}
An inner product on this space van be defined by
\begin{equation}
\left<f,g\right>=\frac{1}{\pi}\int_\C f(z)g^c(z) e^{-|z|^2}dz, 
\end{equation}
Alternatively, the inner product for $f$ and $g=\sum_{n=1}^{\infty} b_n z^n$ in $\tilde \cH_\cB$ is given by 
\begin{equation}
\left<f,g\right>=\sum_{n=1}^{\infty} a_n b_n^c n!,
\end{equation}
and functions in $\tilde \cH_\cB$ satisfy
\begin{equation}
\left<f,g\right>=\sum_{n=1}^{\infty} |a_n|^2 n!<\infty,
\end{equation}
This is the Fock (or Bargmann-Fock, or Segal-Bargmann) space of real analytic (entire) functions \cite{Bargmann:1962,NewmanandShapiro:1966,carswelletal:2003}\footnote{I am thankful to Professor Mihai Putinar for directing my attention to the Fock space.}. 
Complex monomials $e_n=z^n/(n!)^{1/2}$ form an orthonormal basis in $\tilde \cH_\cB$, which is a Reproducing Kernel Hilbert Space (RKHS) with kernel
\be
K(z,w)=K_z(w)=e^{z\cdot w}.
\ee 
The functions corresponding to Taylor series (\ref{eq:Tay}) with real coefficients $a_n$ form a linear subspace of $\tilde \cH_\cB$. Any real function can then be represented by (\ref{eq:Tay})
with $z=x+iy=x,$ i.e. restriction of the complex function to the real line.

We complexify (\ref{eq:one}) to the whole complex plane by analytic continuation, since $F$ is entire.
Let $S^t$ be the flow of complexified (\ref{eq:one}) and assume it is entire. We are interested in whether $f(S^t)$ is in $\tilde \cH_\cB$, i.e. 
\begin{equation}
\int_\C |f(S^tz)|^2 e^{-|z|^2}dz<\infty
\end{equation}
We have 
\begin{equation}
\int_\C |f(S^tz)|^2 e^{-|z|^2}dz=\int_\C |f(y)|^2 e^{-|S^{-t}y|^2}\frac{dS^{-t}y}{dy}dy,
\end{equation}
for $y=S^t z$. As an example, let $F(z)=\lambda z,\lambda<1$. Then 
\begin{equation}
\int_\C |f(S^tz)|^2 e^{-|z|^2}dz=e^{-\lambda t}\int_\C |f(y)|^2 e^{-|e^{-\lambda t }y|^2}dy<\infty,
\end{equation}
for any $t$, provided $\lambda<0$. Thus $\tilde \cH_\cB$ is invariant for linear stable systems. 

However, Theorem 1 in \cite{carswelletal:2003} implies that the case of linear transformation (assuming the globally stable fixed point is at $0$) is the only case in which the Fock space is closed under conjugacy. 
We assume there is an entire conjugacy $h$ to a linear system.
It is then of interest to use the conjugacy $h:\bC\rar\bC$  to define the Modulated Fock Space (MFS), as the space of all the entire functions $g=f\circ h$ on $\C$ vanishing at $0$ such that 
\begin{equation}
\int_\C |f\circ h|^2 e^{-|h|^2}dh<\infty,
\label{MFS}
\end{equation}
The conjugacy $h$ is provided by the principal  eigenfunction \cite{MohrandMezic:2014} of the complexified system at eigenvalue $\lambda$. It is then clear that monomials in $h$
belong to MFS. Namely, choosing $f(h)=h^n$, the integral in  (\ref{MFS}) converges.
\begin{remark} In contrast, note that the integral 
\begin{equation}
\int_\C h^2(z) e^{-|z|^2}dz,
\label{counterMFS}
\end{equation}
might not converge, even if $h$ is in the Fock space. 
\end{remark} 
We have the following:
\begin{proposition} 
Let $h:\bC\rar \bC$ be an entire principal eigenfunction of the complexified flow $S^t$ globally stable to the origin. Assume $h$ is real on the real line in $\C$. Then the MFS is an RKHS space with respect to the kernel  
\be
K(z,w)=K_z(w)=e^{h(z)\cdot h(w)}\ , z,w\in \bC.
\ee 
and the Koopman operator on it is closed under composition with the flow $S^t$.
\label{prop:1dMFS}
\end{proposition}
\begin{proof}
The statement on the RKHS is clear. The closedness under the flow  is proven as follows:
Consider
\begin{eqnarray}
\int_\C |f\circ h(S^t(z))|^2 e^{-|h(z)|^2}dh(z)&=&\int_\C |f(h)|^2 e^{-|h(S^{-t}z)|^2}\frac{dS^{-t}z}{dz}d(h\circ S^{-t} (z))\nonumber \\
&=&\int_\C |f(h(z))|^2 e^{-|e^{-\lambda t}h(z))|^2}d(e^{-\lambda t}h(z))\nonumber \\
&=&e^{-\lambda t}\int_\C |f(h)|^2 e^{-|e^{-\lambda t}h|^2}dh\nonumber \\
&<&\infty.
\end{eqnarray}
\end{proof}
Note that, again due to Theorem 1 in \cite{carswelletal:2003}, the MFS is not the same as the underlying Fock space except in the trivial case when the conjugacy $h$ is linear, and thus identity. MFS is a space of functions on which the globally stable dynamical system has the simple spectrum $n\lambda, n\in \bN$.
\end{example}
We generalize the above example as follows:
\begin{theorem} Consider the flow  $S^t$ of a dynamical system  in $\bR^n$ with a globally stable Milnor attractor $\caA$ of Sell type. Let $\mu$ be the invariant physical measure on $\mu$. Let $\s=(s_1,...,s_n)$ be principal eigenfunctions of $U^t$ that vanish on $\caA$ and correspond to negative real part eigenvalues ${\bm \lambda}=(\lambda_1,...,\lambda_n)$. Let MFS be the space of entire functions on $\bC^n$ such that 
\begin{equation}
\int_\C |f\circ \s|^2 e^{-|\s|^2}d\s<\infty,
\end{equation}
which is the RKHS with respect to 
\be
K(\z,\w)=K_z(w)=e^{\s(\z)\cdot \s(\w)}\ , \ \ \z,\w\in \bC^n.
\ee 
Let $\cH_\caA=L^2(\mu)$ and $\tilde \cH_\cB$ be the MFS. Then the spectrum of the Koopman operator $U^t$ is given by 
\bea
\sigma(U^t)&=&\cl({P(\sigma(U^t|_{\cH_\caA}),\sigma(U^t|_{\cH_\cB})))}\nonumber \\
                     &=&\cl({P(\sigma(U^t|_{\cH_\caA}),\k\cdot{\bm \lambda}))}, \k=(k_1,...,k_n)\in \bN^n
                     \label{eq:spectSell}
\eea
\label{MFSspectrum}
\end{theorem}
\begin{proof} The first line in (\ref{eq:spectSell})  is a direct consequence of Theorem \ref{the:product}, where we only need to check that $U^t$ is closed under composition with the flow on the MFS in this, multi-dimensional case. This is an easy extension of the proof in one-dimension provided in Proposition \ref{prop:1dMFS}. The second line is due to the fact that all the monomials $\s^\k=s_1^{k_1}\cdot...\cdot s_n^{k_n}$ are in the MFS and are eigenfunctions of $U^t$.
\end{proof}
\begin{remark} Restriction of the functions in MFS to the ones defined  on $\bR^n$ in the above theorem yields the same spectrum, as complex monomials are real-valued on $\bR^n$.
\end{remark}
}
\section{Spectral Expansion for Limit Cycling Systems in $\R^2$}
\index{Spectral expansion}
\label{sect:spectlc2D}
Consider a $C^2$ vector field $\F(\x)$ in  $\R^2$ that has a stable limit cycle with domain of attraction ${\cal B}$. According to \cite{LanandMezic:2013}, the two-dimensional set of ordinary differential
equations can - inside the domain of attraction of the limit cycle -  be transformed to the following form:
\begin{eqnarray}
\dot y&=& A(s)y, \nonumber \\ 
\dot s&=&1,
\label{eq:lc0}
\end{eqnarray}
where $y \in \mathbb{R}, s\in \mathbb{S}^1$ and $A(s)$ is a $2\pi$-periodic function. We have the following theorem:

\begin{theorem}\footnote{\color{black} The author proved the result a while ago  and presented it in an early draft manuscript (UCSB preprint, 2011) with Dr. Yuehang Lan.
 That manuscript never got completed or published. The announcement of the result can be found in \cite{MohrandMezic:2014}, where function spaces were built based on the intuition from this example. Since only the statement and no proof was provided in \cite{MohrandMezic:2014},
the author felt it appropriate to provide the complete result with the proof here. 
} Any function $F(y,s)$ that is analytic in $y$ and $L^2$ in $s$ can be expanded into eigenfunctions of the Koopman operator associated with (\ref{eq:lc0}) as follows:
\be
F(y,s)=\sum_{m=0}^\infty\sum_{n=-\infty}^\infty a_{mn} y^me^{-m\int_0^s(A(\bar s)-A^*)d\bar s}e^{ins},
\ee
where $a_{mn}$ are constant,
\be
A^*=\frac{1}{2\pi}\int_0^{2\pi}A(z)dz.
\ee and 
\be
y^me^{-m\int_0^s(A(\bar s)-A^*)d\bar s)}e^{ins}
\ee
is an eigenfunction of the Koopman operator corresponding to the eigenvalue
\be
e^{(mA^*+in)t}.
\ee
\label{ther:lc}
\end{theorem}

We first prove the following lemma:
\begin{lem}The system (\ref{eq:lc0}) has $C^1$ eigenfunctions of the Koopman family of operators that are of form
\be
g(y,s)=b(s)y^\alpha,
\ee
where the function $b(s)$ is $2\pi$-periodic and of the form
\begin{eqnarray}
b(s)&=&ce^{-\alpha a^f(s)},
\mbox{ where }  c  \mbox{ is an arbitrary constant,}\\
a^f(s)&=&\int_0^s (A(\bar s)-A^*)d\bar s , \mbox{ and } \\
A^*&=&\frac{1}{2\pi}\int_0^{2\pi}A(z)dz.
\end{eqnarray}
corresponding to the eigenvalue
$
\mu(t)=e^{\alpha A^*t}.
$
\end{lem}
\begin{proof}
We need to find functions that satisfy
\be
g(y(t),s(t))=\mu(t)g(y_0,s_0).
\ee
Observing from (\ref{eq:lc0}) that
\begin{eqnarray}
s(t)&=&s_0+t, \\
y(t)&=&y_0e^{\int_0^t A(s_0+\bar t)d\bar t},
\end{eqnarray}
we get
\be
g(y_0e^{\int_0^t A(s_0+\bar t)d\bar t},s_0+t)=\mu(t)g(y_0,s_0).
\label{eq:KElc}
\ee
By taking the derivative with respect to time and setting $t=0$, we obtain the partial differential equation
\be
\frac{\p g}{\p y}(y_0,s_0)A(s_0)y_0+\frac{\p g}{\p s}(y_0,s_0)=\dot \mu(0)g(y_0,s_0),
\label{pde}
\ee
whose solution is
\be
g(y_0,s_0)=b(s_0)y_0^\alpha
\label{form}
\ee
where $\alpha$ is a real number and $b$ is a function that satisfies
\be
\frac{db}{ds}(s_0)=(\dot\mu(0)-\alpha A(s_0))b(s_0),
\label{b'}
\ee
as can be verified directly by plugging (\ref{form}) into (\ref{pde}). The function $b$ must be periodic to be $C^1$. Solving (\ref{b'}), we get
\be
b(s_0)=ce^{\dot \mu(0)s_0}e^{-\alpha\int_0^{s_0}A(s)ds},
\label{b}
\ee
 We now show $\dot \mu(0)=\alpha A^*$. From (\ref{b'}) we get
 \be
 \frac{db/ds}{b}(s_0)=(\dot\mu(0)-\alpha A(s_0))
 \ee
 Since $b$ is periodic in $s_0$, integrating both sides with respect to $s_0$ from $0$ to $2\pi$ leads to
\be
 2\pi \dot\mu(0)=\int_0^{2\pi}\alpha A(s_0)ds_0.
\ee
Going back to (\ref{eq:KElc}) and plugging in the expression for $g$ that we just obtained, we get $\mu(t)=e^{\alpha A^*t}$, as claimed.
\end{proof}

\begin{proof} of the Theorem.
It is clear that $e^{int}$ is an eigenvalue of the Koopman operator associated with the eigenfunction $e^{ins}$, since
\be
U^te^{ins}=e^{in(s+t)}=e^{int}e^{ins}
\ee
Now, if two functions $f_1$ and $f_2$ are eigenfunctions of the  Koopman operator with two eigenvalues $e^{at}$ and $e^{bt}$, 
then so is $f=f_1\cdot f_2$, with eigenvalue $e^{(a+b)t}$. Thus, any function of the form 
\be
y^me^{-m\int_0^s(A(\bar s)-A^*)d\bar s}e^{ins}
\label{ev}
\ee
is an eigenfunction of the Koopman operator corresponding to the eigenvalue
$e^{(mA^*+in)t}.$
We have to prove now that any function $F(y,s)$ analytic in $y$ and $L^2$ in $s$ can be expanded as a countable sum of terms 
in (\ref{ev}). We first expand such an $F(y,s)$ in Taylor series in $y$:
$$
F(y,s)=\sum_{m=0}^\infty a_m(s)y^m
$$
where $a_m(s)$ is an $L^2$, $2\pi$ periodic  function of $s$. 
Since 
$$
e^{-m\int_0^s(A(\bar s)-A^*)d\bar s}
$$
is a strictly positive, bounded function of $s$, we can write
$$
a_m(s)=e^{-m\int_0^s(A(\bar s)-A^*)d\bar s}\bar a_m(s)
$$
Expanding 
$$
\bar a_m(s)=\sum_{n=-\infty}^{\infty} a_{mn} e^{ins}
$$
completes the proof.
\end{proof}

Assume now that the transformation $\h(\x)=(y(\x),s(\x))$ is such that $\h(\x)$ is analytic. Then,
for any vector-valued function $$\G(\x):\R^2\rar\R^m,$$ analytic in $\x$, we have 
\be
\G(\x)=\tilde \G(y(\x),s(\x))=\sum_{m=0}^\infty\sum_{n=-\infty}^\infty \a_{mn} y(\x)^me^{-m\int_0^{s(\x)}(A(\bar s)-A^*)d\bar s}e^{ins(\x)},
\ee
where $$y(\x)^me^{-m\int_0^{s(\x)}(A(\bar s)-A^*)d\bar s}e^{ins(\x)}$$ is an eigenfunction of $U^t$ associated with the vector field $\F(\x)$ and thus the time-evolution of $\G$ can be written as
\be
U^t\G(\x)=\sum_{m=0}^\infty\sum_{n=-\infty}^\infty \a_{mn} e^{(mA^*+in)t}y(\x)^me^{-m\int_0^{s(\x)}(A(\bar s)-A^*)d\bar s}e^{ins(\x)},
\ee
where $\a_{mn}$ are the Koopman modes associated with $\G$.  \index{Koopman mode! planar limit cycling system}
Just like in the case of linear systems and nonlinear systems with a stable equilibrium, the expansion retains the form of the sum of products of Koopman eigenvalues, Koopman eigenfunctions and Koopman modes. The eigenvalues and eigenfunctions are the property of the system, and do not change in the expansion from one observable to another. The Koopman modes are associated with a specific observable.
\begin{rem} The above theorem indicates an interesting property of spectral expansions that carries over to more general attractors: smoothness requirements are different for on-attractor evolution (in this case, a limit cycle)   and off-attractor evolution. On the attractor, the appropriate functional space is $L^2$, while in directions tranverse to it ($y$ in the above example) at least continuity (and usually more, e.g. analyticity) is required.
\end{rem}
{\color{black}
\begin{rem} The emphasis in Theorem \ref{ther:lc} is on spectral expansion and not on the spectrum itself. The convergence of the spectral series to the function $F(y,s)$ is pointwise in $y$, and in $L^2$ in $s$. The spectral expansion can be pursued without utilizing all the elements of the  spectrum. As an example, consider the irrational rotation on $S^1$, equiped with the Haar measure $\mu$,  $\theta\rar \theta+2\pi\omega, \theta\in[0,2\pi)$, where $\omega$ is an irrational number in $[0,1)$. The eigenvalues of the associated unitary Koopman operator acting on functions in $L^2(\mu)$, $Uf(\theta)=f(\theta+2\pi\omega)$ are given by $e^{in2\pi\omega}, \ n\in \bZ$, while its spectrum is the closure of the set of eigenvalues that are dense in $S^1$, i.e. $\sigma (U)=S^1$. The spectral expansion for any $f\in L^2(\mu)$
reads $f(\theta)=\sum_{n\in\bZ} a_n e^{in\theta}$, and thus 
\be
Uf=\sum_{n\in\bZ} a_n e^{in2\pi\omega}e^{in\theta}.
\label{spectexprot}
\ee
Contrasting with the spectral theorem for unitary operators, that states 
\be
Uf=\int_{[0,1)}e^{i2\pi\omega}dP_\omega(f)
\ee
we see that the spectral expansion in (\ref{spectexprot}) utilizes a dense subset of the spectrum instead of the whole spectrum.
\end{rem}}
\color{black}
\begin{corollary}
\label{cor:r}
Any $C^2$, two dimensional system of ordinary differential equations globally asymptotically stable to a limit cycle can be written 
in the form
\bea
\dot r&=&A^* r \\
\dot\theta&=&\omega.
\eea
\end{corollary}
\begin{proof}
Let $r=e^{-\int_0^s(A(\bar s)-A^*)d\bar s}y$.
\end{proof}
{\color{black} We can now connect the above result in Theorem \ref{ther:lc} with Theorem \ref{dissHilb}, by setting up appropriate function spaces and discussing the spectrum itself. Let $\cD$ be the set $r\leq D$, where $r$ is defined in Corollary \ref{cor:r}, and $D$ is a positive constant. Let $\cH_\caA=L^2(\caA,\mu)=L^2(s)$, the space of functions of $s$ that are square integrable. Here we will take a slightly different approach than  in Theorem \ref{MFSspectrum}, since the construction is more direct and does not involve prior knowledge of principal eigenfunctions. The space $\tilde\cH_\cB$ can be constructed as a space carrying some properties of a Reproducing Kernel Hilbert Space. In particular, we construct it as a space of power series vanishing at $y=0$, whose indeterminates are monomials in $y \in \C$, while the coefficients are square integrable functions of $s$ (see \cite{MohrandMezic:2014} for such a construction in the Banach space context):
\be
F(y,s)=\sum_{m=1}^\infty f_m(s)y^m, \ f_m(s)\in L^2(s).
\ee
such that 
\be
\frac{1}{2\pi}\sum_{m=1}^\infty m!   \int_{S^1}f^2_m(s)ds<\infty.
\ee
We equip $\tilde \cH_\cB$ with the inner product 
\be
\left< F,G\right>=\frac{1}{2\pi}\sum_{m=1}^\infty m!  \int_{S^1}f_m(s)g_m(s)ds,
\ee
If we define a kernel of the form
\be
K(y,z)=\sum_{m=1}^\infty \frac{y^m z^m}{m!},
\ee
we evaluate 
\be
K_y(z)=K(\cdot,y)(z)=\sum_{m=1}^\infty \frac{y^m z^m}{m!}=\sum_{m=1}^\infty k_m(s)z^m,
\ee
the coefficients $k_m$ are given by
\be
k_m(s)=\frac{y^m}{m!} \in L^2(s)
\ee
i.e. they are constant functions on the circle.
We also have
\be
\left< K_z(y),F\right>=\frac{1}{2\pi}\sum_{m=1}^\infty   y^m\int_{S^1}f_m(s)ds=\frac{1}{2\pi}\int_{S^1}F(s,y)ds=\bar F(s,y).
\ee
Given this characterization, we call $\tilde\cH_\cB$ the Averaging Kernel Hilbert Space (AKHS).

It is now evident, using Theorem \ref{dissHilb}, and the proof of the Theorem \ref{ther:lc} that the following is true:
\begin{cor} Assume $A(s)\leq0$. The spectrum of $U^t$ in $\tilde\cH$ is $\cl(\{mA^*+in, m\in\mathbb{N},m\in\mathbb{Z}\} ).$
\end{cor}
\begin{proof} We only need to prove that AKHS is closed under composition with $S^t$, which follows from 
\be
F(y(t,y_0),s(t,s_0))=\sum_{m=1}^\infty f_m(s_0+t)y_0^me^{m\int_0^t A(s_0+\bar t)d\bar t}.
\ee
We have 
\be
\sum_{m=1}^\infty m!|f_m(s_0+t)e^{m\int_0^t A(s_0+\bar t)d\bar t}|^2\leq \sum_{m=1}^\infty m! |f_m(s_0+t)|^2<\infty
\ee
and thus, $F(y(t,y_0),s(t,s_0))$ is in AKHS for every $t$.
\end{proof}

\begin{remark} The AKHS defined above can also be defined as the space of functions $f(s,z),s\in S^1, z\in \C,$ such that 
\begin{equation}
\int_{S^1}\int_\C |f(s,z)|^2 e^{-|z|^2}dzds<\infty.
\end{equation}
\end{remark}
\color{black}
\section{Spectral Expansion for  Limit Cycling Systems in $\R^n$}
\label{sect:lc}
Consider now a $C^2$ vector field $\F(\x)$ in  $\R^{n+1},$ that has a stable limit cycle with domain of attraction ${\cal B}$. The corresponding  set of ordinary differential
equations $\dot\x=\F(\x)$ can - inside the domain of attraction of the limit cycle -  be transformed to the following form:\footnote{Note the increase of dimension of state space by $1$, to $n+1,$ to accommodate easier notation since it makes the off-attractor vector $\y$ $n$-dimensional.}
\begin{eqnarray}
\dot \y&=& A(s)\y, \nonumber \\ 
\dot s&=&1,
\label{eq:lcn}
\end{eqnarray}
where $\y \in \bR^n, s\in \mathbb{S}^1$ and $A(s)$ is a $2\pi$-periodic matrix. We have the following theorem:

\begin{ther} Let Floquet exponents $\bm{\mu}=(\mu_1,...,\mu_n)$ of (\ref{eq:lcn}) (eigenvalues of the Floquet stability matrix $B$) be distinct. Any vector-valued  function $\G(\y,s)$ that is analytic in $\y$ and $L^2$ in $s$ can be expanded into eigenfunctions of the Koopman operator associated with (\ref{eq:lcn}) as follows:
$$
\G(\y,s)=\sum_{\m\in \bN^n,k\in \bZ} \a_{\m k} \z^\m(\y,s) e^{iks},
$$
where $\a_{\m k}$ are constant Koopman modes, $\m=(m_1,...,m_k)$, 
$$\z(\y,s)=(z_1(\y,s),...,z_k(\y,s))$$
 are principal Koopman eigenfunctions defined by $$\z(\y,s)=V^{-1}P^{-1}(s)\y,$$ where $P(s)$ is the Floquet periodic matrix, $V^{-1}$ is the diagonalizing matrix for $B$, and $$\z^\m=z_1^{m_1}(\y,s)\cdot...\cdot z_k^{m_k}(\y,s)e^{iks}$$
is the eigenfunction of the Koopman operator corresponding to the eigenvalue
$$
e^{(\m\cdot\bm{\mu}+ik)t}.
$$
\label{ther:spectlcn}
\end{ther}
\begin{proof}

From Floquet theory \cite{Hale:1969} we know that the transformation $\y=P(s)\tilde \y$ leads to
$$
\dot{\tilde \y}=B\tilde\y,
$$
where $P(s)$ is the Floquet periodic matrix $B$ the $s$-independent stability matrix.
In  case of the matrix $B$ with $n$ independent eigenvectors, let $V$ be such that $V^{-1}BV=\Lambda,$ where $\Lambda$ is a diagonal matrix. Then by setting
\be\z(\y,t)=V^{-1}\tilde\y=V^{-1}P^{-1}(t)\y,
\label{eq:KEFloq}\ee  we
obtain $\z(\y,s)$ as the Koopman eigenfunctions of the suspended system.

Then, we  expand $\G(\y,s)$ in Taylor series to obtain
\bea
\G(\y,s)&=&\sum_{\m\in \bN^n}  b_{\m }(s) \y^\m\nonumber \\
&=& \sum_{\m\in \bN^n}b_{\m }(s) (P(s)V\z(\y,s))^\m\nonumber \\
&=& \sum_{\m\in \bN^n}\bar a_{\m }(s)\z^\m(\y,s)\nonumber \\
&=& \sum_{\m\in \bN^n,k\in \bZ}a_{\m k}e^{iks}\z^\m(\y,s)\nonumber 
\eea
where $\bar a_{\m k}(s)$ are the Koopman modes.\index{Koopman modes! n-dimensional limit cycling system} \end{proof}
In the previous section we were able to explicitly write the spectral expansion for the planar, limit-cycling system in terms of an integral of a scalar function $A(s)$.
In the $n$-dimensional case with a stable limit cycle we  can not exhibit the eigenfunctions explicitly. This is essentially due to non-commutativity of $A(s_1)$ and $A(s_2),$ in the case when $A$'s are matrices. Nevertheless, using a bit of Floquet theory, we obtained a useful expansion that leads to time evolution of the
observable given by 
\be
U^t\G(\y,s)=\sum_{\m\in \bN^n,k\in \bZ}e^{(\m\cdot\bm{\mu}+ik)t}a_{\m k}e^{iks}\z^\m(\y,s).
\label{eq:KElcobs}
\ee
Note that the eigenvalues $$\lambda_{\m,k}=(\m\cdot\bm{\mu}+ik)$$
again form a lattice in the complex plane (see example \ref{exa:lc}).

{\color{black} \begin{rem} We now again construct, $\cH_\caA=L^2(\caA,\mu)=L^2(s)$, the space of functions of $s$ that are square integrable. The space $\tilde\cH_\cB$ can be constructed similarly to the construction of the Averaging Kernel Hilbert Space in the previous section. We construct it as a space of power series vanishing at $\z=0$, whose indeterminates are monomials in $\z\in\C^n$, while the coefficients are square integrable functions of $s$ 
\be
F(y,s)=\sum_{\k, |\k|\geq1} f_\k(s)\z^\m, \ f_\k(s)\in L^2(s), .
\ee
where $|\k|=\sum_{j=1}^n k_j$, and $\z^\m=z_1^{m_1}...z_n^{m_n}$.
We equip $\tilde \cH_\cB$ with the inner product 
\be
\left< F,G\right>=\frac{1}{2\pi}\sum_{\k, |\k|\geq 1} a_\k  \int_{S^1}f_\k(s)g_\k(s)ds,
\ee
where $a_\k=\k!=k_1!k_2!...k_n!$.
If we define a kernel of the form
\be
K(y,z)=\sum_{\k, |\k|\geq 1}\frac{(\z  \w)^\k}{a_\k},
\ee
where $(\z  \w)^\k=(z_1w^1)^{k_1}...(z_nw^n)^{k_n}$ we evaluate 
\be
K_\z(\w)=K(\cdot,\z)(\w)=\sum_{\k, |\k|\geq 1} \frac{(\z \w)^\k}{a_\k}=\sum_{\k, |\k|\geq 1} j_\k(s)\w^\k,
\ee
the coefficients $j_\k$ are given by
\be
j_\k(s)=\frac{\z^\k}{a_\k} \in L^2(s)
\ee
i.e. they are constant functions on the circle.
We also have
\be
\left< K_\z(\w),F\right>=\frac{1}{2\pi}\sum_{\k, |\k|\geq 1}   \z^\k\int_{S^1}f_\k(s)ds=\frac{1}{2\pi}\int_{S^1}F(s,\z)ds=\bar F(s,\z).
\ee

The spectrum of $U^t$ is then characterized by  
\begin{cor} The spectrum of $U^t$ in $\tilde\cH$ is $\cl(\{\lambda_{\m,k}, \m\in\mathbb{N}^n,k\in\mathbb{Z}\} ).$
\end{cor}
\end{rem}
}

\section{Spectral Expansion for Quasiperiodic Attractors\index{Spectral expansion, quasiperiodic} in $\R^n$}
\label{sect:lt}
Consider again  $\dot \x=\F(\x),\x\in \R^n$ that has a quasi-periodic attractor - an $m$-dimensional torus on which the dynamics is conjugate to  $$\dot \btheta =\bomega, \btheta\in {\mathbb T}^m, \bomega\in \bR^m,$$ where $\bomega$ is a constant and incommensurable vector of frequencies that satisfies \be{\bf k}\cdot \bomega\geq c/|{\bf k}|^\gamma\label{eq:Dioph}\ee for some $c,\gamma>0$. In addition, we ask that the quasi periodic linearization matrix $A(\btheta)$  has a full spectrum, where the spectrum $\sigma(A)$ of the quasi-periodic matrix\index{Spectrum, of a quasi-periodic matrix}  is defined as a set of points $\lambda\in\bR$ for which the shifted equation 
$$
\dot \y=(A(\btheta+\bomega t)-\lambda I)\y,
$$
does not have an exponential dichotomy\index{Exponential dichotomy}. Provided $\sigma(A)$ is full - meaning it consists of $m$ isolated points, and $A(\btheta+\bomega t)$ is sufficiently smooth in $\btheta$,  there is a quasi-periodic transformation $P(t)$ and a constant matrix $B$ - which we will cal quasi-Floquet- such that the transformation $\z=P(t)\y$ reduces the system to $\dot \z=B\z$ \cite{Sell:1981}. Analogous to the case with an attracting limit-cycle,  we consider the skew-linear\index{Skew-linear} system\footnote{Note on the terminology: these types of systems were classically called linear skew-product systems.}
\begin{eqnarray}
\dot \y&=& A(\btheta)\y, \nonumber \\ 
\dot \btheta&=&\bomega,
\label{eq:qpn}
\end{eqnarray}
where $\y \in \bR^n, \btheta\in \mathbb{T}^m$ and $A(\btheta)$ is a $2\pi$-periodic matrix. 
\begin{ther} Let quasi-Floquet exponents $\bm{\mu}=(\mu_1,...,\mu_{n-m})$ of (\ref{eq:qpn}) (eigenvalues of the quasi-Floquet matrix $B$) be distinct. Any vector-valued  function $\G(\y,\btheta)$ that is analytic in $\y$ and $L^2$ in $\btheta$ can be expanded into eigenfunctions of the Koopman operator associated with (\ref{eq:qpn}) as follows:
$$
\G(\y,\btheta)=\sum_{\m\in \bN^{k},\k\in \bZ^m} \a_{\m \k} \z^\m(\y,\btheta) e^{i\k\cdot\btheta},
$$
where $\a_{\m \k}$ are the constant Koopman modes, $\m=(m_1,...,m_{k})$, 
$$\z(\y,\btheta)=(z_1(\y,\btheta),...,z_k(\y,\btheta))$$
 are the principal Koopman eigenfunctions defined by $$\z(\y,\btheta)=V^{-1}P^{-1}(\btheta)\y,$$ $V^{-1}$ is the diagonalizing matrix for $B$, and $$\z^\m=z_1^{m_1}(\y,\btheta)\cdot...\cdot z_1^{m_k}(\y,\btheta)e^{i\k\cdot\btheta}$$
is an eigenfunction of the Koopman operator corresponding to the eigenvalue
$
e^{(\m\cdot\bm{\mu}+i\k\cdot\bomega)t}.
$
\label{ther:spectqpn}
\end{ther}
\begin{proof} The proof is analogous to the proof of the theorem \ref{ther:spectlcn}.
\end{proof}
From the above theorem, we have the expression for the evolution of an observable, $L^2$ in $\theta$ and analytic in $\y$ that reads
\be
U^t\G(\y,\btheta)=\sum_{\m\in \bN^{k},\k\in \bZ^m} e^{(\m\cdot\bm{\mu}+i\k\cdot\bomega)t}\a_{\m \k} \z^\m(\y,\btheta) e^{i\k\cdot\btheta}.
\label{eq:KEqpobs}
\ee

\begin{remark}[Spectral expansion for nonlinear  systems with limit cycle or toroidal attractors] \index{Spectral expansion for nonlinear, quasi-periodic systems}
In this chapter we derived spectral expansions corresponding to skew-linear systems that arise through conjugacies with linearization around a (quasi)periodic attractor of a nonlinear system and are extended through the whole basin of attraction using the methodology we first developed in theorem \ref{ther:pullback}. It is interesting to note that the spectral expansion for systems with equilibrium, for example equation (\ref{evol_observable}) are possible in that form only under the non-resonance conditions on the eigenvalues. However, the spectral expansion theorems for skew-linear systems only have a non-resonance condition  - equation (\ref{eq:Dioph}) - in the case when there are multiple angular variables. For full spectral expansion of a nonlinear system with a quasi-periodic attractor, we need the conjugacy to be analytic in off-attractor coordinates, and thus the non-resonance conditions such as those in the Poincar\'e linearization theorem \ref{ther:PLT} would be required.
\end{remark}
{\color{black} The construction of an appropriate Hilbert space in which the spectrum of $U^t$ is the closure of the set of eigenvalues
$\m\cdot\bm{\mu}+i\k\cdot\bomega$ can be done along the lines of the previous section, with the only difference being that the averaging is now over an $m$-dimensional torus, rather than over a circle.}
\subsection{Isostables in Systems with a (Quasi)Periodic Attractor}
In theorems \ref{ther:spectqpn} and \ref{ther:spectlcn} we find special eigenfunctions of the Koopman operator that correspond to important partitions of the state space. For example, in the case of an attracting limit cycle, the level sets of $e^{is(\x)}$ define isochrons. Similarly, for 
\bea 
\m&=&(0,...,0),\nonumber \\
\k&=&(0,...,1_j,...0),
\eea 
the eigenfunction corresponding to a quasi-periodically attracting limit cycle whose level sets define  generalized isochrons is $e^{i\theta_j(\x)}$. 
\index{Isostable! quasi-periodic attractor}
Extending our discussion of linear and nonlinear systems with equilibria, we can also define the notion of generalized isostables through level sets of Koopman eigenfunctions. Let the eigenvalues -(quasi)-Floquet exponents - of the constant matrix $B$
be such that $|\mu_k|>...>|\mu_1|>0$ and their real parts $\sigma_n<...,\sigma_1<0$. Then the level sets of the eigenfunction $z_1(\y(\x),\btheta(\x))$ are defined to be isostables associated with the system. They have the property that initial conditions on such an isostable converge simultaneously to 
the attractor of the system, at the rate $\sigma_1$.
\subsection{Stable, Unstable and Center Manifolds in Systems with (Quasi)Periodic Invariant sets}
Akin to the discussion of stable, unstable and center manifolds in section \ref{nonlman} for systems with equilibria, we can define such manifolds for systems with (quasi)periodic invariant sets (not necessarily attracting), by utilizing spectral expansions such as (\ref{eq:KElcobs}) and (\ref{eq:KEqpobs}), if they exist. Namely,
we select  the principal eigenfunctions - those with $\m=(0,...,1_j,...,0)$. Then we obtain the stable manifold as the joint zero level set of principal eigenfunction corresponding to eigenvalues $\mu_j$ with real part $\sigma_j\geq 0$. Similar consideration leads to definition of center and unstable manifolds.

\section{Principal Coherent Dimension of the Data}
\label{sec:data}
In each case of the dynamical systems in $\R^n$ we studied,the spectrum was found to be of the {\it lattice type} where eigenvalues are given as combinations of $m+j$ {\it principal} eigenvalues,
$$
\lambda_{\n,\k}=\n\cdot\bm{\mu}+i\k\cdot\bomega,
$$
where $\n=(n_1,...,n_m)\in \bN^m$,   $\k=(k_1,...,k_j)\in \bZ^j$, $\bm{\mu}=(\mu_1,...,\mu_m)\in \bC^m$ and $\omega=(\omega_1,...,\omega_m)\in \bR^{+j}$, $j$ is the dimension of the attractor, and $m=n-j$. This leads us to say that data has the {\it principal spectrum} and that the principal coherent dimension of the data is $n$ {\color{black} if the experimentally or numerically observed spectrum has such structure in which $n$ principal eigenvalues generate the rest of the point spectrum.}
\begin{example}
\label{exa:lc}
Consider the three-dimensional, limit cycling system
\begin{eqnarray}
\dot x&=&y  \label{eq:1}\\
\dot y&=&x-x^3-cy  \label{eq:2}\\
\dot \theta&=&\omega. \label{eq:3}
\end{eqnarray}
The two fixed points of the equations (\ref{eq:1}-\ref{eq:2}) are $y=0,x=\pm1$. The linearization matrix at those is 
\be
A=\left[ \begin{array}{cc} 0 & 1 \\ -2 & -c \end{array}\right],
\ee
and thus the eigenvalues are determined by
\be
(-\lambda)(-c-\lambda)+2=\lambda^2+c\lambda+2=0,
\ee
leading to
\be
\lambda_{1,2}=\frac{-c\pm\sqrt{c^2-8}}{2}
\ee
For $c=\sqrt{7},\omega=1$, the eigenvalues  read $\lambda_{3,4}=-1.3228756\pm .5i$. Setting $\omega=1$, the other two principal eigenvalues are $\pm i$.
In figure \ref{EvaLat} we show a subset of the eigenvalues of the Koopman operator on $L^2(S^1)\times {\cal A}$, where ${\cal A}$ is the space of analytic functions on the plane, in the basin of attraction of either of the limit cycles (since they are symmetric) of (\ref{eq:1}-\ref{eq:3}).
\begin{figure}[h!!] 
\centering \hspace{.5cm} \includegraphics[width=11.5cm,
height=10.5cm, clip=true, trim=0 100 0 100]{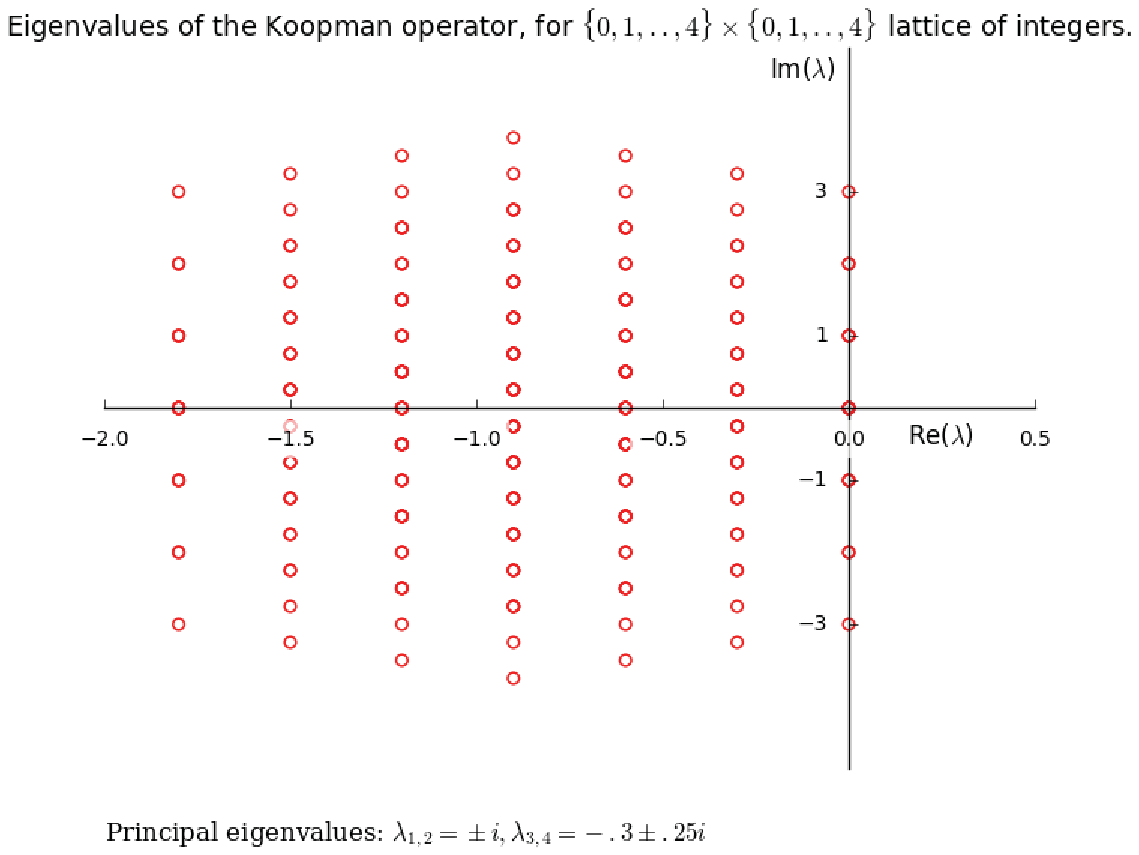}
\caption{ Eigenvalues of the Koopman operator, for $\{0,1,..,4\}\times \{0,1,..,4\}$ lattice of integers.}
\label{EvaLat}
\end{figure}
\end{example}
{\color{black} The lattice spectrum can be approximated from data utilizing, for example, variants of Dynamic Mode Decomposition \cite{Schmid:2010} or compactification methods \cite{giannakisetal:2018}. See for example \cite{crnjaricetal:2017} for some such computations, from which the principal dimension can be determined.}
\section{A Dynamical System with Continuous Spectrum: a Cautionary Tale in Data Analysis}
\label{cont}
In this section we shift away from the dissipative dynamics and consider a measure-preserving system that presents us with an example of a Koopman operator with continuous spectrum. 

While integrable systems are in some sense the simplest nontrivial examples (harmonic oscillator is the simplest, but somewhat trivial) of measure-preserving dynamics, there are already some complexities when considering them from spectral perspective of the Koopman operator. Consider a $1$ degree of freedom system in action-angle variables $(I,\theta)\in {\cal I}\times S^1$, where ${\cal I}=[a,b]\subset \R^+$, given by
\bea
\dot I&=&0, \nonumber \\
\dot \theta &=&I.
\eea
Such a system is produced by e.g. pendulum dynamics in the part of the state space separating oscillating motion from the rotational motion of the pendulum \cite{Wiggins:1990}),
 While $I$ is an eigenfunction at $0$ of the Koopman operator $U^t$, there are no eigenvalues  at any other point on the unit circle, since the associated  eigenfunction $\phi_\omega(I,\theta)$ would have to satisfy 
\be
 U^t\phi_\omega(I,\theta)=\phi_\omega(I,\theta+It)=e^{i\omega t}\phi_\omega(I,\theta).
 \label{eq:eigaa}
\ee
Write $\phi_\omega=re^{i\varphi(I,\theta)}$ since the modulus $r$ does not change with time, to obtain 
\be
 \phi_\omega(I,\theta+It)=re^{i(\varphi(I,\theta)+It)}=re^{i(\varphi(I,\theta)+\omega t)}.
\ee
However, this can be satisfied only for $\omega=I$, and thus the eigenfunction in a proper sense does not exist. 

Now, we could define 
\be
\phi(I,\theta)=e^{i\theta}\delta(I-c)
\ee
 as the Dirac delta function defined on $I=c$, a closed circle in the $I-\theta$ plane, and that ``function" would satisfy (\ref{eq:eigaa}), in a weak sense: 
\bea
\int_{M}U^t\phi(I,\theta)w(I,\theta)dId\theta&=&\int_{M}\phi(I,\theta+It)w(I,\theta)dId\theta\nonumber \\
&=&\int_{M}e^{i(\theta+It)}\delta(I-c)w(I,\theta)dId\theta \nonumber \\
&=&e^{ict}\int_{M}\phi(I,\theta)w(I,\theta)dId\theta. 
\eea
where $w$ is a smooth, compactly supported function on ${\cal I}\times S^1$.
However,  $\phi$ is  not a function, but a measure. 
In fact, there is a family of measures
$
\phi_j(I,\theta)=e^{ij\theta}\delta(I-c),
$
that satisfy
\bea
\int_{M}U^t\phi_j(I,\theta)w(I,\theta)dId\theta&=&\int_{M}\phi_j(I,\theta+It)w(I,\theta)dId\theta\nonumber \\&=&\int_{M}e^{ij(\theta+It)}\delta(I-c)w(I,\theta)dId\theta \nonumber \\
&=&e^{ijct}\int_{M}\phi_j(I,\theta)w(I,\theta)dId\theta, 
\eea
that we call eigenmeasures.\index{Eigenmeasure}
It turns out that the the Koopman operator for the above equation has {\it continuous spectrum}. The continuous spectrum can be understood as the extension of the notion of the point spectrum, but for which eigenfunctions are replaced by eigenmeasures.

Consider now a square integrable function $f(I,\theta)$. Its evolution under the Koopman operator is given by 
$$
U^tf(I,\theta)=f(I,\theta+It).
$$
Expanding $f$ into Fourier series we obtain
$$
f(I,\theta)=\sum_{j\in \bZ} a_j(I)e^{ij\theta},
$$
and thus 
\be
U^tf(I,\theta)=f(I,\theta+It)=\sum_{j\in \bZ} e^{ijIt}a_j(I)e^{ij\theta}.
\label{eq:evoaa}
\ee
We will show that the evolution of $f$ has a spectral expansion
\be
U^tf(I,\theta)=f^*(I)+\int_{\bR}e^{i\beta t}dP_{\beta}(f(I,\theta)),
\label{eq:contspectev}
\ee
where $f^*(I)=a_0(I)$ is the time average of $f$ along trajectories and $dP_\beta$ is the ``differential" of the  so-called projection valued measure on $\bR$  which is a map from Borel sets on the real line  to the set of all linear projection operators on the set of square integrable functions. Recall that a linear operator is a projection\index{Projection operator} if it satisfies $P^2=P$, i.e. applying it twice we get the same result as applying it once. Now, let $\bZ'=\bZ-\{0\},$ $f'=f-f^*,$ so $f'$ has zero mean, and define 
\be
dP_\beta(f(I,\theta))=\sum_{j\in \bZ'} a_j(I)e^{ij\theta}\delta(jI-\beta)d\beta,
\ee
where $\beta$ is the Lebesgue measure on $\bR$. The projection valued measure $P$ is then defined by
$
P(A)=\int_A dP_\beta,
$
for any Borel set $A$ in Borel $\sigma$-algebra on $\bR$. With this, it should become clear why we called $dP_\beta$ the ``differential" of the projection-valued measure. To show that $P$ is  a projection valued measure, we need to show that, when evaluated on the full set $\bR$,  it is equal to identity on the Hilbert space of square integrable functions of zero mean, i.e. $P(\bR)=I$ and, in addition, that 
$\mu(A)=\int_M P(A)g^cf dId\theta$ is a measure on $\bR$, where $f,g$ are both of zero mean.

Firstly, note that 
\be
P(\bR)(f'(I,\theta))=\int_\bR dP_\beta(f'(I,\theta))=\int_\bR \sum_{j\in \bZ'} a_j(I)e^{ij\theta}\delta(jI-\beta)d\beta
=\sum_{j\in \bZ'} a_j(I)e^{ij\theta}=f'(I,\theta),\ee
and thus $P(\bR)$ is identity on the space of zero-mean functions $\bR$. 
Secondly, integration against a function $g(I,\theta)=\sum_{k\in \bZ'} b_k(I)e^{ik\theta}$ gives:
\be
\int_M g^c(I,\theta)dP_\beta(f(I,\theta))dI d\theta=  2\pi\sum_{k\in \bZ'} \kappa(\beta)b_k(\beta) a_k(\beta)d\beta, 
\ee
where $\kappa(\beta)$ is a function that gives a (finite, bounded) integer number of times $\beta=jI$ where $j$ is an integer and $I\in{\cal I}.$
The last expression is a differential of a measure on $\bR$,
and $ \sum_{k\in \bZ'} \kappa(\beta)b_k(\beta) a_k(\beta)$ is a square integrable function. Therefore, we get  an absolutely continuous measure
\be
\mu(A)=\int_A 2\pi\sum_{k\in \bZ'} \kappa(\beta)b_k(\beta) a_k(\beta)d\beta.
\ee
We have
\be
\int_{\bR}e^{i\beta t}dP_{\beta}(f(I,\theta))=\int_{\bR}e^{i\beta t}\sum_{j\in \bZ'} a_j(I)e^{ij\theta}\delta(jI-\beta)d\beta
=\sum_{j\in \bZ'}e^{ijIt} a_j(I)e^{ij\theta} 
\ee
and since $a_0(I)=f^*(I)$ this proves our assertion (\ref{eq:contspectev}). 

The expression for the evolution of a function under the action of the Koopman operator is quite interesting to consider from the perspective of an experimentalist\index{Experiment! pendulum}. Say one is studying the motion of a mechanical pendulum governed by the equation for the angle $\rho$ and angular velocity $\psi$
\bea
\dot \rho&=&\psi, \nonumber \\
\dot \psi&=&-\frac{g}{l}\sin\rho,
\eea
and assume the initial condition $(\rho_0,\psi_0)$ is in  the region of state space inside the separatrix, where action-angle coordinates can be defined. The experiment could for example be performed by taking a video of the motion and extracting the angular position $\rho$ by image processing. According to our theory, Fourier analysis of the  evolution of observable $\rho$ will show peaks at frequencies $jI_0(\rho_0,\psi_0)$, where $I_0(\rho_0,\psi_0)$ is the action corresponding to initial condition $(\rho_0,\psi_0)$. Changing the initial condition to $(\rho_1,\psi_1)$ will lead to a different peaked spectrum where peaks are at frequencies $jI_1(\rho_1,\psi_1)$. The point here is that, while the spectrum of the Koopman operator is continuous in the sense discussed above, measurement of the spectrum from a single trajectory will in this case lead to a peaked spectrum with generally different peaks associated with initial conditions of different actions.  Since in experiment a single initial condition is chosen, the only consequence of the continuous spectrum to this experimental situation is that relevant frequencies change continuously\index{Spectrum! continuous} with initial conditions.  Koopman operator spectrum and spectrum obtained from a single initial condition could coincide in cases when dynamics is more complicated than that of a pendulum, e.g. in the case when the dynamics on the attractor is mixing.   But, in general, the spectrum is expected to be the same for almost all initial conditions in the same ergodic component, and different for initial conditions in different ergodic components \cite{Mezic:2016}. 
\section{Conclusions}
In this paper, spectral expansions are derived for a class of dynamical systems possessing pure point spectrum for observables 
that are $L^2$ on the attractor and analytic in off-attractor directions. {\color{black} Hilbert spaces are constructed utilizing the tensor product of  on-attractor and off-attractor Hilbert  spaces, rendering the Koopman operator spectral. It is interesting to note that natural Hilbert spaces for dissipative linear systems- Fock spaces - are not closed when the dynamics is nonlinear. This might help explain the phenomenon that occurs when polynomials are used as EDMD basis \cite{PageandKerswell:2019}, where the spectrum ``switches" as a trajectory of the system goes from one of the singular solutions such as an unstable fixed point, towards another. The intuition drown from present work is that the function space is unstable, and the dynamics ``leaks out" of the subspace it was in. Also,  it is typical in prior approaches to defining Koopman operator analysis in RKHS's (see e.g. \cite{Kawahara:2016,Klusetal:2017}) to skirt the issue of the closedness of the RKHS under the action of a dynamical system.  Thus again, the dynamics might ``leak" out of the defined RKHS space. Using principal eigenfunctions, the Modified Fock Space is created here that does not encounter that problem, as the dynamics stays within an invariant subspace. The Hankel DMD methods (see e.g. \cite{ArbabiandMezic:2017}) also have good invariance properties and have been proven useful in prediction tasks \cite{khodkaretal:2018}. }. 

The notion of generalized eigenfuctions of the composition (Koopman) operator is utilized to derive a spectral decomposition of linear and skew-linear systems. The concept of open eigenfunctions is defined, and used together  with conjugacy to (skew)-linear systems  to construct eigenfunctions in stable, unstable and saddle-point equilibria cases. Consequences for the geometry of the state-space for such systems are derived by considering level sets of specific eigenfunctions. 
Notably, the concepts of stable, unstable and center manifolds are redefined using joint level sets of (generalized) eigenfunctions. The analysis is extended to the case of (quasi)-periodic attractors, and the appropriate spectral expansions derived, and the concept of isostables extended for such systems. The results in this paper largely carry over, with appropriate modifications (for example in the resonance conditions) for discrete-time maps. An example of a measure-preserving, integrable system with a continuous spectrum is presented, motivating the discussion of the type of spectrum of the Koopman operator of a system vs. the type of spectrum computed for time-evolution of data from a single initial condition of the same system. Finally, the discussion of types of spectrum found in data is related to the spectrum of analyzed dynamical systems, enabling identification of dynamical system type directly from data.

{\bf Acknowledgements} I am thankful to John Guckenheimer, Dimitris Giannakis, Mihai Putinar, Yueheng Lan, Alex Mauroy, Ryan Mohr and Mathias Wanner for their useful comments. This work was supported in part by the DARPA contract HR0011-16-C-0116 and ARO grants W911NF-11-1-0511 and W911NF-14-1-0359.

\label{sect:conc}
\bibliographystyle{unsrt}
\bibliography{MOTDyS,alex,KvN,darpabdd}

\end{document}